\documentclass[a4paper,10pt]{llncs}

\usepackage{newtxmath}
\usepackage{newtxtext}

\usepackage{siunitx}

\usepackage[utf8]{inputenc}
\usepackage{graphicx}
\usepackage{amsmath}
\usepackage{amssymb}
\usepackage{scalerel}
\usepackage{stackrel}
\usepackage{textcomp}
\usepackage{extarrows}
\usepackage{stmaryrd}
\usepackage{enumerate}
\usepackage{thm-restate}
\usepackage{paralist}
\usepackage{cleveref}
\usepackage{subcaption}
\usepackage{verbatim}
\usepackage{mathtools}
\usepackage[inline,shortlabels]{enumitem}
\usepackage{booktabs}
\usepackage{yfonts}
\usepackage[T1]{fontenc}
\usepackage{wrapfig}
\usepackage{xspace}
\usepackage{tikzsymbols}
\usepackage{trimclip}
\usepackage[font={small}]{caption}
\usepackage{blindtext}

\usepackage{pgf}
\usepackage{tikz}
\usetikzlibrary{automata,positioning,trees,shapes,petri,arrows,snakes,backgrounds,calc,fit,arrows.meta,shadows,math,decorations.markings}
\usetikzlibrary{shapes.multipart,arrows,automata}
\usetikzlibrary{shapes.geometric}
\usetikzlibrary{positioning}

\pgfdeclarelayer{mybackground}
\pgfdeclarelayer{background}
\pgfdeclarelayer{foreground}
\pgfsetlayers{background,mybackground,foreground,main}

\tikzstyle{background rectangle}=
[rounded corners,bottom color=green!10,top color=green!10,inner frame sep=0pt]
\tikzstyle{fitnode}=[draw=black, rounded corners,line width=0.5pt,dotted]
\tikzstyle{statenode}=
[circle,fill=blue!10,draw,inner sep=0.1pt, minimum size=3.2mm,text=blue!50!black,font=\footnotesize]
\tikzstyle{sstatenode}=
[circle,fill=blue!10,draw,inner sep=0.1pt, minimum size=7mm,text=blue!50!black,font=\footnotesize]
\tikzstyle{lblnode}=
[rounded corners,fill=red!10,inner sep=2pt,font=\footnotesize]
\tikzstyle{autedge}=[->,line width=0.5pt]
\tikzstyle{labelnode}=[text=black,align=left,sloped,above=-2pt,font=\footnotesize]
\tikzstyle{commentnode}=
[ellipse callout,draw,fill=green!10,inner sep=0.5pt, font=\footnotesize]
\tikzstyle{blocknode}=
[rectangle,rounded corners,draw,minimum width=50pt,minimum height=25pt,line width=2pt,align=center]
\tikzstyle{dotnode}=
[circle,fill=black,minimum size=6pt,inner sep=0pt]
\tikzstyle{cndnode}=
[kite,kite upper vertex angle=120,,kite lower vertex angle=120,draw,line width=2pt]
\tikzstyle{storenode}=
[ellipse,draw,line width=2pt,align=center]
\tikzstyle{blockedge}=[line width=2pt,->,>=stealth]
\tikzstyle{storeedge}=[dotted,line width=2pt,->,>=stealth]

\usepackage[linesnumbered,lined,commentsnumbered,ruled]{algorithm2e}
\InputIfFileExists{synctex.tex}{%
\synctex=1
}{%
}

\InputIfFileExists{cutmargins.tex}{%
	\pdfpagesattr{/CropBox [125 170 490 730]} 
}{%
}
\usetikzlibrary{external}

\input{macros.tex}

\author{
  Vojt\v{e}ch Havlena,
  Luk\'{a}\v{s} Hol\'{i}k,
  Ond\v{r}ej Leng\'{a}l, and
  Tom\'{a}\v{s} Vojnar
}

\institute{
  FIT,
  IT4I Centre of Excellence,
  Brno University of Technology,
  Czech Republic}

\newcommand{\inmaintext}[0]{}

\title{
Automata Terms in a Lazy WS$k$S Decision Procedure\\
(Technical Report)
}

\begin{document}

\maketitle

\begin{abstract}

We propose a lazy decision procedure for the logic \wsks.
It builds a~term-based symbolic representation of the state space of the tree
automaton (TA) constructed by the classical \wsks decision procedure.
The classical decision procedure transforms the symbolic representation into a~TA
via a~bottom-up traversal and then tests its language non-emptiness, which
corresponds to satisfiability of the formula.
On the other hand, we start evaluating the representation from the top,
construct the state space on the fly, and utilize opportunities to prune away
parts of the state space irrelevant to the language emptiness test.
In order to do so, we needed to extend the notion of \emph{language terms}
(denoting language derivatives) used in our previous procedure for the linear
fragment of the logic (the so-called WS1S) into \emph{automata terms}.
We implemented our decision procedure and identified classes of formulae on
which our prototype implementation is significantly faster than the classical
procedure implemented in the \mona tool.

\end{abstract}

\section{Introduction}\label{sec:intro}

\emph{Weak monadic second-order logic of $k$ successors} (\wsks) is a logic for
describing regular properties of finite $k$-ary trees.
In addition to talking about trees, \wsks can also encode complex properties
of a~rich class of general graphs by referring to their tree backbones~\cite{pale01}.
\wsks offers extreme succinctness for the price of non-elementary worst-case
complexity.
%
As noticed first by the authors of \cite{glenn-wia-96} in the context
of WS1S (a~restriction that speaks about finite words only), the trade-off
between complexity and succinctness may, however, be turned significantly favourable in
many practical cases through a~use of clever implementation techniques and
heuristics.
Such techniques were then elaborated in the tool \mona{}
\cite{monapaper,monamanual}, the best-known implementation of decision
procedures for WS1S and WS2S.
\mona{} has found numerous applications in verification of programs with complex
dynamic linked data structures \cite{pale01,strand1,strand2,hip/sleek,jahob},
string programs \cite{tateishi}, array programs \cite{zhou}, parametric systems
\cite{baukus,bodeveix,RaduBeeping2}, distributed systems
\cite{KlaNieSun:casestudyautver,smith_00}, hardware verification \cite{basin},
automated synthesis \cite{sandholm,hune,hamza}, and even computational
linguistics \cite{linguistics97}.

Despite the extensive research and engineering effort invested into \mona, due
to which it still offers the best all-around performance among existing
WS1S/WS2S decision procedures, it is, however, easy to reach its scalability
limits.
Particularly, \mona{} implements the classical WS1S/WS2S decision procedures
that build a word/tree automaton representing models of the given formula and
then check emptiness of the automaton's language.
The non-elementary complexity manifests in that the size of the automaton is
prone to explode, which is caused mainly by the repeated determinisation (needed
to handle negation and alternation of quantifiers) and synchronous product
construction (used to handle conjunctions and disjunctions).
Users of \wsks are then forced to either find workarounds, such as
in~\cite{strand2}, or, often restricting the input of their approach, give up
using \wsks altogether~\cite{kuncak:trex}.

As in \mona{}, we further consider WS2S only (this does not change the
expressive power of the logic since $k$-ary trees can be easily encoded into
binary ones).
We revisit the use of tree automata (TAs) in the WS2S decision procedure and
obtain a new decision procedure that is much more efficient in certain cases.
It is inspired by works on \emph{antichain algorithms} for efficient testing of
universality and language inclusion of finite automata
\cite{doyen:antichain,wulf:antichains,bouajjani:antichain,abdulla:when}, which
implement the operations of testing emptiness of a~complement (universality) or
emptiness of a~product of one automaton with the~complement of the other one
(language inclusion) via an \emph{on-the-fly} determinisation and product
construction.
The on-the-fly approach allows one to achieve significant savings by pruning the
state space that is irrelevant for the language emptiness test.
The pruning is achieved by early termination when detecting non-emptiness (which
represents a simple form of \emph{lazy evaluation}), and \emph{subsumption}
(which basically allows one to disregard proof obligations that are implied by
other ones).
Antichain algorithms and their generalizations have shown great efficiency
improvements in applications such as abstract regular model checking
\cite{bouajjani:antichain}, shape analysis \cite{habermehl:forest}, LTL model
checking \cite{antiLTL08}, or game solving \cite{antiGames06}.

Our work generalizes the above mentioned approaches of on-the-fly automata
construction, subsumption, and lazy evaluation for the needs of deciding WS2S.
In our procedure, the TAs that are constructed explicitly by the classical
procedure are represented symbolically by the so-called \emph{automata terms}.
More precisely, we build automata terms for subformulae that start with
a~quantifier (and for the top-level formula) only---unlike the classical
procedure, which builds a~TA for every subformula.
Intuitively, automata terms specify the set of leaf states of the TAs of the
appropriate (sub)formulae.
The leaf states themselves are then represented by \emph{state terms}, whose
structure records the automata constructions (corresponding to Boolean
operations and quantification on the formula level) used to create the given TAs
from base TAs corresponding to atomic formulae.
The leaves of the terms correspond to states of the base automata.
Automata terms may be used as state terms over which further automata terms of
an even higher level are built.
Non-leaf states, the transition relation, and root states are then given
implicitly by the transition relations of the base automata and the structure~of~the~state~terms.

Our approach is a generalization of our earlier work \cite{gaston} on WS1S.
Although the term structure and the generalized algorithm may seem close to
\cite{gaston}, the reasoning behind it is significantly more involved.
Particularly, \cite{gaston} is based on defining the semantics (language) of
terms as a function of the semantics of their sub-terms.
For instance, the semantics of the term $\{q_1,\ldots,q_n\}$ is defined as the
union of languages of the state terms $q_1,\ldots,q_n$, where the language of a
state of the base automaton consists of the words \emph{accepted at that state}.
With TAs, it is, however, not meaningful to talk about trees accepted from a leaf
state, instead, we need to talk about a given state and its \emph{context},
i.e., other states that could be obtained via a bottom-up traversal over the
given set of symbols.
Indeed, trees have multiple leafs, which may be accepted by a number of different
states, and so a tree is \emph{accepted from a set of states}, not from any single
one of them alone.
We therefore cannot define the semantics of a state term as a tree language, and
so we cannot define the semantics of an automata term as the union of the
languages of its state sub-terms.
This problem seems critical at first because without a sensible notion of the
meaning of terms, a straightforward generalization of the algorithm of
\cite{gaston} to trees does not seem possible.
The solution we present here is based on defining the semantics of terms via the
automata constructions they represent rather then as functions of languages of
their sub-terms.

Unlike the classical decision procedure, which builds a~TA corresponding to
a~formula \emph{bottom-up}, i.e. from the atomic formulae, we build automata terms
\emph{top-down}, i.e., from the top-level formula.
This approach offers a lot of space for various optimisations.
Most importantly, we test non-emptiness of the terms \emph{on the fly} during
their construction and construct the terms \emph{lazily}.
In particular, we use \emph{short-circuiting} for dealing with the $\wedge$ and
$\vee$ connectives and \emph{early termination} with possible
\emph{continuation} when implementing the fixpoint computations needed when
dealing with quantifiers.
That is, we terminate the fixpoint computation whenever the emptiness can be
decided in the given computation context and continue with the computation when
such a need appears once the context is changed on some higher-term level.
Further, we define a notion of \emph{subsumption} of terms, which, intuitively,
compares the terms wrt the sets of trees they represent, and allows us to
discard terms that are subsumed by others.

We have implemented our approach in a prototype tool.
When experimenting with it, we have identified multiple parametric families of
WS2S formulae where our implementation can---despite its prototypical
form---significantly outperform \mona{}.
We find this encouraging since there is a~lot of space for further
optimisations and, moreover, our implementation can be easily combined with
\mona{} by treating automata constructed by \mona in the same way as if
they were obtained from atomic predicates.

This is an extended version of the paper~\cite{lazy-cade}.

\section{Preliminaries}\label{sec:prelims}

In this section, we introduce basic notation, trees, and tree automata,
and give a~quick introduction to the \emph{weak monadic
second-order logic of two successors} (WS2S) and its classical decision
procedure.
We give the minimal syntax of WS2S only; see, e.g., Comon \emph{et
al.}~\cite{tata} for more details.

\subsubsection{Basics, Trees, and Tree Automata.}


Let $\Sigma$ be a finite set of symbols, called an \emph{alphabet}.
The set $\Sigma^*$ of \emph{words} over $\Sigma$ consists of finite
sequences of symbols from $\Sigma$.
The \emph{empty word} is denoted by $\epsilon$, with $\epsilon \not \in
\Sigma$.
The \emph{concatenation} of two words~$u$ and~$v$ is denoted by $u.v$ or
simply~$uv$.
The \emph{domain} of a partial function $f: X \to Y$ is the set
%
%
$\domof{f} = \{x \in X \mid \exists y: x\mapsto y \in f\}$, its \emph{image}
is the set
%
%
$\imgof{f} = \{y \in Y \mid \exists x: x\mapsto y \in f \}$, and its
\emph{restriction} to a~set~$Z$ is the function $\frestr{Z} = f \cap
(Z\times Y)$.
For a~binary operator $\bullet$, we write $A \timesof{\bullet} B$ to denote the
augmented product $\{a \bullet b \mid (a,b)\in  A \times B\}$ of $A$ and $B$.


We will consider ordered binary trees.
We call a~word $p\in\leftright^*$ a~tree \emph{position} and $p.\tleft$ and
$p.\tright$ its \emph{left} and \emph{right child}, respectively.
Given an~alphabet~$\Sigma$ s.t.~$\bot \notin \Sigma$, a~\emph{tree}
over~$\Sigma$ is a finite partial function $\tree: \leftright^* \rightarrow
(\Sigma \cup \{\bot\})$ such that
\begin{inparaenum}[(i)]
\item $\domof{\tree}$ is non-empty and prefix-closed, and
\item for all positions $p \in \domof{t}$, either $\tree(p)\in\Sigma$ and $p$
has both children, or $\tree(p) = \bot$ and $p$ has no children, in which case
it is called a~\emph{leaf}.
We let $\leav(\tree)$ be the set of all leaves of $\tree$.
\end{inparaenum}
The position $\epsilon$ is called the \emph{root}, and we
write~$\alltreesx{\Sigma}$ to denote the set of all trees over~$\Sigma$\footnote{%
Intuitively, the $\alltreesx{[\cdot]}$ operator can be seen as a~generalization
of the Kleene star to tree languages.
The symbol~\!\!$\scalebox{1.4}{\chinmu}$ is the Chinese character for a~tree,
pronounced \emph{m\`{u}}, as in English \emph{moo-n}, but shorter and with
a~falling tone, staccato-like.
}.
We abbreviate $\alltreesx{\{a\}}$ as $\alltreesx{a}$ for $a\in\Sigma$.

The \emph{sub-tree} of $\tree$ rooted at a~position $p\in\domof{\tree}$ is the
tree $\tree' = \{p'\mapsto \tree(p.p')\mid p.p'\in\domof{\tree}\}$.
A \emph{prefix} of $\tree$ is a~tree $\tree'$ such that
$\restr{\tree'}{\domof{\tree'}\setminus\leav(\tree')}\subseteq
\restr{\tree}{\domof{\tree}\setminus\leav(\tree)}$.
The \emph{derivative} of a tree $\tree$ wrt a set of trees $S
\subseteq\alltrees$ is the set $\tree- S$ of all prefixes $\tree'$ of $\tree$
such that, for each position $p\in\leav(\tree')$, the sub-tree of $\tree$ at $p$
either belongs to $S$ or it is a leaf of $\tree$.
Intuitively, $\tree-S$ are all prefixes of $\tree$ obtained from $\tree$ by
removing some of the sub-trees in~$S$.
The derivative of a~set of trees~$T \subseteq \alltreesx \Sigma$ wrt~$S$ is the
set $\bigcup_{\tree \in T}(\tree - S)$.


A~(binary) \emph{tree automaton} (TA) over an alphabet~$\Sigma$ is a~quadruple
$\A = (Q, \trans, \leafs, \roots)$ where $Q$ is a~finite set of \emph{states},
$\delta: Q^2 \times \Sigma \to 2^Q$ is a~\emph{transition function},
$\leafs\subseteq Q$ is a~set of \emph{leaf states}, and $\roots\subseteq Q$ is
a~set of \emph{root} states.
We use $(q, r) \ltr{a} s$~to~denote~that $s \in \trans((q,r), a)$.
A~\emph{run} of~$\A$ on~a tree $\tree$ is a total map $\rho: \domof \tree \to Q$
such that
if $\tree(p) = \bot$, then $\rho(p) \in \leafs$,
else $(\rho(p.\tleft), \rho(p.\tright)) \ltr{a} \rho(p)$ with $a = \tree(p)$.
The run~$\rho$ is \emph{accepting} if $\rho(\epsilon) \in \roots$, and the
\emph{language} $\langof \A$ of~$\A$ is the set of all trees on which
$\A$~has~an~accepting run.
$\A$ is \emph{deterministic} if~$|I| = 1$ and $\forall q,r \in Q, a \in \Sigma:
|\delta((q,r), a)| \leq 1$, and \emph{complete} if $I \geq 1$ and $\forall q,r
\in Q, a \in \Sigma: |\delta((q,r), a)| \geq 1$.
Last, for $a \in \Sigma$, we shorten $\delta((q,r),a)$ as $\delta_a(q,r)$, and
we use $\delta_\Gamma(q,r)$ to denote $\bigcup\{\delta_a(q,r) \mid a \in
\Gamma\}$ for a~set~$\Gamma \subseteq \Sigma$.
%

%
%
%


\subsubsection{Syntax and Semantics of WS2S.}\label{sec:label}

\begin{figure}[t]
  \raisebox{41mm}{
  \begin{subfigure}[t]{0.44\linewidth}
  \begin{center}
    \scalebox{0.95}{
\begin{tikzpicture}

\path[use as bounding box] (-29mm,5mm) rectangle (29mm,-40mm);

\ifx\inmaintext\undefined
\newcommand{\tleft}{\mathtt{L}}
\newcommand{\tright}{\mathtt{R}}
\fi

\tikzset{level 1/.style={sibling distance=25mm,level distance=10mm}}
\tikzset{level 2/.style={sibling distance=13mm,level distance=10mm}}
\tikzset{level 3/.style={sibling distance=7mm,level distance=10mm}}
\tikzset{level 4/.style={sibling distance=4mm,level distance=6mm}}

\tikzstyle{emptynode}=[draw,circle,fill=black,inner sep=0.2mm]
\tikzstyle{bluenode}=[emptynode,inner sep=0.8mm,draw=blue,fill=blue!50]
\tikzstyle{voidnode}=[emptynode,draw=none,fill=none]
\tikzstyle{noedge}=[dashed]
\tikzstyle{infiniteedge}=[dotted]

\node (epsilon) [emptynode,label=left:{$\epsilon$}]  {}
  child {node (L) [emptynode,label=left:{$\tleft$}] {}
    child {node (LL) [voidnode] {} edge from parent[noedge]
      child {node (LLL) [voidnode] {} edge from parent[noedge]
        child {node (LLLL) [voidnode] {} edge from parent[infiniteedge]}
        child {node (LLLR) [voidnode] {} edge from parent[infiniteedge]}
      }
      child {node (LLR) [voidnode] {} edge from parent[noedge]
        child {node (LLLL) [voidnode] {} edge from parent[infiniteedge]}
        child {node (LLLR) [voidnode] {} edge from parent[infiniteedge]}
      }
    }
    child {node (LR) [bluenode,label=left:{$\tleft\tright$}] {}
      child {node (LRL) [voidnode] {} edge from parent[noedge]
        child {node (LLLL) [voidnode] {} edge from parent[infiniteedge]}
        child {node (LLLR) [voidnode] {} edge from parent[infiniteedge]}
      }
      child {node (LRR) [voidnode] {} edge from parent[noedge]
        child {node (LRRL) [voidnode] {} edge from parent[infiniteedge]}
        child {node (LRRR) [voidnode] {} edge from parent[infiniteedge]}
      }
    }
  }
  child {node (R) [bluenode,label=left:{$\tright$}] {}
    child {node (RL) [emptynode,label=left:{$\tright\tleft$}] {}
      child {node (RLL) [voidnode] {} edge from parent[noedge]
        child {node (RLLL) [voidnode] {} edge from parent[infiniteedge]}
        child {node (RLLR) [voidnode] {} edge from parent[infiniteedge]}
      }
      child {node (RLR) [bluenode,label=left:{$\tright\tleft\tright$}] {}
        child {node (RLRL) [voidnode] {} edge from parent[infiniteedge]}
        child {node (RLRR) [voidnode] {} edge from parent[infiniteedge]}
      }
    }
    child {node (RR) [bluenode,label=left:{$\tright\tright$}] {}
      child {node (RRL) [voidnode] {} edge from parent[noedge]
        child {node (RRLL) [voidnode] {} edge from parent[infiniteedge]}
        child {node (RRLR) [voidnode] {} edge from parent[infiniteedge]}
      }
      child {node (RRR) [voidnode] {} edge from parent[noedge]
        child {node (RRRL) [voidnode] {} edge from parent[infiniteedge]}
        child {node (RRRR) [voidnode] {} edge from parent[infiniteedge]}
      }
    }
  }
;

\end{tikzpicture}
    }
  \end{center}
  \vspace*{-2mm}
  \vspace*{-4mm}
  \caption{Positions assigned to the variable $X$.
    }
  \label{fig:positions}
  \end{subfigure}
  }
  %
    \hfill
    \begin{subfigure}[t]{0.52\linewidth}
    \begin{center}
      \scalebox{0.95}{
\begin{tikzpicture}

\usetikzlibrary{shapes.multipart}

\tikzset{level/.style={sibling distance=30mm/#1,level distance=10mm}}
\tikzstyle{emptynode}=[draw,circle,fill=black,inner sep=0.2mm]
\tikzstyle{bluenode}=[emptynode,inner sep=0.8mm,draw=blue,fill=blue!50]
\tikzstyle{voidnode}=[emptynode,draw=none,fill=none]

\ifx\inmaintext\undefined
\newcommand{\tleft}{\mathtt{L}}
\newcommand{\tright}{\mathtt{R}}

\tikzstyle{varnode}=[emptynode,rectangle split, rounded corners, rectangle split parts=2,
  rectangle split horizontal,rectangle split part align=base,inner sep=0.8mm, draw, align=center,fill=none]
\tikzstyle{varnodeno}=[varnode,rectangle split part fill={black!15,black!15}]
\newcommand{\parted}[2]{#1\nodepart{two}#2}

\fi

\tikzstyle{varnodefirst}=[varnodeno,rectangle split part fill={blue!50,black!15}]
\tikzstyle{varnodesecond}=[varnodeno,rectangle split part fill={black!15,red!50}]
\tikzstyle{varnodeboth}=[varnodeno,rectangle split part fill={blue!50,red!50}]
\tikzstyle{bottomish}=[level distance=6mm]
\tikzstyle{botnode}=[inner sep=0.1mm]

\node (epsilon) [varnodesecond,label=left:{$\epsilon$}] {\parted 0 1}
  child {node (L) [varnodesecond,label=left:{$\tleft$}] {\parted 0 1}
    child {node (LL) [varnodesecond,label=left:{$\tleft\tleft$}] {\parted 0 1}
      child[bottomish] {node (LLL) [botnode] {$\bot$}}
      child[bottomish] {node (LLR) [botnode] {$\bot$}}
    }
    child {node (LR) [varnodefirst,label=left:{$\tleft\tright$}] {\parted 1 0}
      child[bottomish] {node (LRL) [botnode] {$\bot$}}
      child[bottomish] {node (LRR) [botnode] {$\bot$}}
    }
  }
  child {node (R) [varnodeboth,label=left:{$\tright$}] {\parted 1 1}
    child {node (RL) [varnodeno,label=left:{$\tright\tleft$}] {\parted 0 0}
      child[bottomish] {node[botnode] (RLL) [botnode] {$\bot$}}
      child {node (RLR) [varnodefirst,label=left:{$\tright\tleft\tright$}] {\parted 1 0}
        child[bottomish] {node (RLRL) [botnode] {$\bot$}}
        child[bottomish] {node (RLRR) [botnode] {$\bot$}}
      }
    }
    child {node (RR) [varnodeboth,label=left:{$\tright\tright$}] {\parted 1 1}
      child[bottomish] {node (RRL) [botnode] {$\bot$}}
      child[bottomish] {node (RRR) [botnode] {$\bot$}}
    }
  }
;

\end{tikzpicture}
      }
    \end{center}
    \vspace*{-4mm}
    \caption{Encoding of~$\asgn$ into a tree~$\treeof \asgn$; a~node at
      a~position~$p$ has the value \tikz[baseline,anchor=base]{\node[varnodeno]
      {\parted{$x$}{$y$}};} where $x = 1$ iff $\treeof{\asgn}(p)$ maps~$X$ to~$1$
      and $y = 1$ iff $\treeof{\asgn}(p)$ maps~$Y$ to~$1$.}
    \label{fig:encoding_example:encoding}
    \end{subfigure}
    \vspace*{-2mm}
\caption{An example of an assignment~$\asgn$ to a~pair of variables~$\{X, Y\}$
  s.t. $\asgn(X) =\{\tleft\tright,
    \tright, \tright\tleft\tright,\tright\tright\}$ and
    $\asgn(Y) = \{\epsilon,\tleft,
     \tleft\tleft, \tright,\tright\tright\}$
    and its encoding into a~tree.}
\label{fig:encoding_example}
\end{figure}

WS2S is a~logic that allows quantification over second-order \emph{variables},
which are denoted by upper-case letters $X, Y, \ldots$
and range over \emph{finite sets} of tree positions in $\leftright^*$ (the
finiteness of variable assignments is reflected in the name \emph{weak}).
%
See Fig.~\ref{fig:positions} for an example of a~set of positions
assigned to a variable.
%
%
Atomic formulae (atoms) of WS2S are of the form:
\begin{inparaenum}[(i)]
\item  $X \subseteq Y$,
%
%
%
\item  $X = \suckleftof Y$, and
\item  $X = \suckrightof Y$.
\end{inparaenum}
Formulae are constructed from atoms using the logical connectives $\land,
%
%
\neg$, and the quantifier~$\exists \vars$ where $\vars$ is a~finite set of
variables (we write $\exists X$ when $\vars$ is a~singleton set $\{X\}$).
Other connectives (such as $\vee$ or $\forall$) and predicates (such as the
predicate $\singof X$ for a singleton set $X$) can be obtained as syntactic
sugar (see App.~\ref{app:wsks-pred}).

A~\emph{model} of a~WS2S formula $\varphi(\vars)$ with the set of free
variables~$\vars$ is an assignment $\asgn: \vars \to 2^{\leftright^*}$ of the
free variables of~$\varphi$ to finite subsets of~$\leftright^*$ for which the
formula is \emph{satisfied}, written $\asgn \models \varphi$.
Satisfaction of atomic formulae is defined as follows:
\begin{inparaenum}[(i)]
\item  $\asgn \models X \subseteq Y$ iff $\asgn(X) \subseteq \asgn(Y)$,
%
%
%
\item  $\asgn \models X = \suckleftof Y$ iff $\asgn(X) = \{p.\tleft \mid p \in
\asgn(Y)\}$, and
\item  $\asgn \models X = \suckrightof Y$ iff $\asgn(X) = \{p.\tright \mid p \in
\asgn(Y)\}$.
\end{inparaenum}
Informally, the $\suckleftof Y$ function returns all positions from~$Y$ shifted
to their left child and the $\suckrightof Y$ function returns all positions from~$Y$
shifted to their right child.
Satisfaction of formulae built using Boolean connectives and the quantifier is
defined as usual.
A~formula~$\varphi$ is \emph{valid}, written $\models \varphi$, iff all
assignments of its free variables are its models, and \emph{satisfiable} if it
has a~model.
Wlog, we assume that each variable in a~formula either has only free occurrences
or is quantified exactly once; we denote the set of (free and quantified)
variables occurring in a~formula~$\varphi$ as~$\varsof \varphi$.

\vspace{-0.0mm}
\subsubsection{Representing Models as Trees.}\label{sec:label}
\vspace{-0.0mm}

We fix a~formula~$\varphi$ with variables~$\varsof \varphi = \vars$.
%
A~\emph{symbol}~$\symb$ over~$\vars$ is a (total)~function $\symb: \vars \to \{0,1\}$,
e.g., $\symb = \{X \mapsto 0, Y \mapsto 1\}$ is a~symbol over $\vars = \{X, Y\}$. 
We use $\alphof \vars$ to denote the set of all symbols over~$\vars$
and~$\zerosymb$ to denote the symbol mapping all variables in $\vars$ to~$0$, i.e.,
$\zerosymb = \{X \mapsto 0 \mid X \in \vars\}$.
%



A~finite assignment~$\asgn: \vars \to 2^{\leftright^*}$ of $\varphi$'s variables
can be encoded as a~finite~tree~$\treeof \asgn$ of symbols over~$\vars$ where
every position~$p \in \leftright^*$ satisfies the following conditions:
\begin{inparaenum}[(a)]
  \item
    if $p \in \asgn(X)$, then $\treeof \asgn(p)$ contains $\{X \mapsto 1\}$, and
  \item
    if $p \notin \asgn(X)$, then either $\treeof \asgn(p)$ contains~$\{X \mapsto
    0\}$ or $\treeof \asgn(p) = \bot$ (note that the occurrences of~$\bot$
    in~$\tree$ are limited since~$\tree$ still needs to be a~tree).
\end{inparaenum}
Observe that~$\asgn$ can have multiple encodings: the unique minimum one~$\minenc$ and
(infinitely many) extensions of~$\minenc$ with $\zerosymb$-only trees.
The \emph{language} of $\varphi$ is defined as the set of all encodings of
its models $\langof\varphi = \{
\treeof{\asgn} \in \alltreesx{\alphof\vars} \mid  \asgn \models \varphi \text{
  and } \treeof{\asgn} \text{ is an encoding of } \asgn\}$.


%
%
%



Let $\symb$ be a~symbol over~$\vars$.
For a set of variables $\varsY \subseteq \vars$, we define the \emph{projection}
of~$\symb$ wrt~$\varsY$ as the set of symbols~$\projof \varsY \symb =
\{\symb' \in \alphof \vars \mid \symb_{|\vars\setminus\varsY} \subseteq
\symb'\}$.
%
Intuitively, the projection removes the original assignments of variables
from~$\varsY$ and allows them to be substituted by any possible value.
We define $\projof \varsY \bot = \bot$ and
write~$\pi_Y$ if~$\varsY$ is a~singleton set~$\{Y\}$.
As an example, for $\vars = \{ X,Y \}$ the projection of $\zerosymb$ wrt
$\{X\}$ is given as $\projof{X}{\zerosymb} = \{ \{ X\mapsto 0, Y\mapsto 0 \},
\{ X\mapsto 1, Y\mapsto 0 \} \}$.\footnote{%
Note that our definition of projection differs from the usual one, which would
in the example produce a~single symbol $\{Y \mapsto 0\}$ over a~different
alphabet (the alphabet of symbols over~$\{Y\}$).
}
The definition of projection can be extended to trees $\tree$ over~$\Sigma_{\vars}$
so that $\projof \varsY \tree$ is the set of trees $\{\tree' \in
\alltreesx{\alphof \vars} \mid
\forall p \in \posof \tree : \text{if }\tree(p) = \bot, \text{ then } \tree'(p) =
\bot, \text{ else } \tree'(p) \in \projof \varsY {\tree(p)} \}$
and subsequently to
languages $L$ so that
$\projof \varsY L = \bigcup\{ \projof \varsY \tree \mid \tree\in L \}$.
\subsubsection{The Classical Decision Procedure for WS2S.}\label{sec:treeaut_cons}
\vspace{-0.0mm}

The classical decision procedure for the WS2S logic goes through a direct
construction of a TA $\autof\varphi$ having the same language as a~given formula $\varphi$.
%
%
Let us briefly recall the automata constructions used (cf.~\cite{tata}).
Given a~complete TA $\A = (Q,\delta,I,R)$,
the \emph{complement} assumes that $\A$ is deterministic and returns $\complof\A = (Q,\delta,I,Q\setminus R)$,
the projection returns $\pi_X(\A) = (Q,\delta^{\pi_X},I,R)$ with
$\delta_a^{\pi_X}(q,r) = \delta_{\projof X a}(q,r)$,
%
and the \emph{subset construction} returns the deterministic and complete automaton
$\subsetof\A =
(2^Q,\subsetof\delta,\{I\},\subsetof R)$ where
$\subsetof{\delta_a}(S,S') = \bigcup_{q\in S,q'\in S'}\delta_a(q,q')$ and
$\subsetof R = \{ S\subseteq Q\mid S\cap R\neq \emptyset \}$.
%
%
The binary operators $\circ\in\{\cup,\cap\}$ are implemented through
a~\emph{product construction}, which---given the TA $\A$ and another complete TA
$\A' = (Q',\delta',I',R')$---returns the automaton $\A\circ\A'=(Q\times
Q',\Delta^\times,I^\times,R^\circ)$ where
$\Delta^\times_a((q,r),(q',r')) = \Delta_a(q,q')\times\Delta'_a(r,r')$,
$I^\times = I\times I'$, and
%
for $(q,r)\in Q\times Q'$,
$(q,r)\in R^\cap \Leftrightarrow q\in R\land r\in R'$ and
$(q,r)\in R^\cup \Leftrightarrow q\in R\lor r\in R'$.
%
%
The language non-emptiness test can be implemented through the equivalence
$\langof{\A}\neq \emptyset$ iff $\reach_{\trans}(I)\cap R\neq \emptyset$
where the set $\reach_{\trans}(S)$ of states \emph{reachable} from a set $S\subseteq Q$ through $\delta$-transitions is computed as the least fixpoint
\begin{equation}
\reach_{\delta}(S)= \mu Z.~S \cup \bigcup_{q, r\in Z} \delta({q},{{r}}) .
\label{eq:reach}
\end{equation}
The same fixpoint computation is used to compute the derivative wrt
$\alltreesx a$ for some $a\in \Sigma$ as
$\A - \alltreesx a = (Q,\delta,\reach_{\delta_a}(I),R)$: the new leaf states are all those reachable from $I$ through $a$-transitions.

The classical \wsks decision procedure uses the above operations to constructs the automaton $\autof\varphi$ inductively to the structure of $\varphi$ as follows:
\begin{inparaenum}[(i)]
\item If $\varphi$ is an atomic formula, then $\autof\varphi$ is a pre-defined
  \emph{base} TA over~$\alphof \vars$
(the particular base automata for our atomic predicates can be found, e.g., in \cite{tata},
and we list them also in App.~\ref{sec:basicTA}).
\item If $\varphi = \varphi_1 \land \varphi_2$,
then
$\autof\varphi = \autof{\varphi_1} \cap \autof{\varphi_2}$.
\item If $ \varphi = \varphi_1 \lor \varphi_2$,
then
$\autof\varphi = \autof{\varphi_1} \cup \autof{\varphi_2}$.
\item If  $\varphi = \neg \psi$, then $\autof\varphi = \complof{\autof\psi}$.
\item Finally, if $\varphi = \exists X.\ \psi$, then
$\autof\varphi = \subsetof{(\pi_X(\autof\psi))}-\allzerotrees$.
\end{inparaenum}

Points (i) to (iv) are self-explanatory.
In point (v), the projection implements the quantification by forgetting the
values of the $X$~component of all symbols.
Since this yields non-determinism, projection is followed by determinisation
by the subset construction.
Further, the projection can produce some new trees that contain $\zerosymb$-only
labelled sub-trees, which need not be present in some smaller encodings of the
same model.
Consider, for example, a~formula $\psi$ having the language~$\langof\psi$ given
by the tree~$\treeof \asgn$ in Fig.~\ref{fig:encoding_example:encoding} and all
its $\zerosymb$-extensions.
To obtain $\langof{\exists X.\psi}$, it is not sufficient to make the
projection $\projof{X}{\langof\psi}$ because the projected language does not
contain the minimum encoding $\minenc$ of $\asgn: Y \mapsto \{\epsilon, \tleft,
\tleft\tleft, \tright, \tright\tright\}$, but only those encodings~$\asgn'$
such that~$\asgn'(\tright\tleft\tright) = \{Y \mapsto 0\}$.
Therefore,
the $\zerosymb$-derivative is needed to saturate the language with \emph{all} encodings
of the encoded models (if some of these encodings were missing, the inductive
construction could produce a wrong result, for instance, if the language were
subsequently complemented).
%
%
Note that the same effect can be achieved by replacing the set of leaf states~$I$
of~$\autof{\varphi}$ by $\reach_{\Delta_{\zerosymb}}(I)$ where $\Delta$ is the
transition function of $\autof \varphi$.
See~\cite{tata} for more details.

\section{Automata Terms} \label{sec:terms}

Our algorithm for deciding WS2S may be seen as an alternative implementation of the classical procedure from \secref{sec:treeaut_cons}.
The main innovation is the data structure of \emph{automata terms}, which implicitly represent the automata constructed by the automata operations.
Unlike the classical procedure---which proceeds by a~bottom-up traversal on the
formula structure, building an automaton for each sub-formula before proceeding
upwards---automata terms allow for constructing parts of automata at higher levels
from parts of automata on the lower levels even though the construction of the
lower level automata has not yet finished.
This allows one to test the language emptiness on the fly
and use techniques of state space pruning, which will be discussed later in
\secref{sec:algorithm}.
Proofs of the lemmas can be found in App.~\ref{app:proofs}.

\makeatletter
\newcommand{\specialcell}[1]{\ifmeasuring@#1\else\omit$\displaystyle#1$\ignorespaces\fi}
\makeatother

\begin{wrapfigure}[8]{r}{5.2cm}
\vspace*{-13mm}
\hspace*{-8mm}
\begin{minipage}{6.6cm}
\begin{flalign*}
A ::={}&\specialcell{ S \mid D  \hfill\textit{(automata term)}}\\
S ::={}&\specialcell{ \{ t, \dots, t \} \hfill\textit{(set term)}}\\
D ::={}&\specialcell{ S - \allzerotrees  \hfill\textit{(derivative term)}}\\
t ::={}&\specialcell{  q \mid t \disj t \mid t \conj t \mid {}
\hspace{6mm}\hfill\raisebox{-2mm}{\textit{(state term)}}} \\[-2mm]
  &\overline t \mid \pi_X(t) \mid S \mid D\hspace{4mm}
\end{flalign*}
  \end{minipage}
	\vspace{-2mm}
	\caption{Syntax of terms.}
	\label{fig:syntax}
\end{wrapfigure}

\vspace{-4mm}
\subsubsection{Syntax of automata terms.}
Terms are created according to the grammar in Fig.~\ref{fig:syntax} starting from
%
%
states~$q \in Q_i$, denoted as \emph{atomic states}, of a~given finite set of \emph{base
automata} $\B_i = (Q_i,\delta_i,I_i,R_i)$ with pairwise
disjoint sets of states.
%
%
%
For simplicity, we assume that the base automata are complete,
%
and we denote by $\B = (Q^\B,\delta^\B,I^\B,R^\B)$ their component-wise union.
\emph{Automata terms}~$A$ specify the set of leaf states of an automaton.
\emph{Set terms}~$S$ list a~finite number of the leaf states explicitly, while
\emph{derivative terms}~$D$ specify them symbolically as states reachable from
a~set of states~$S$ via~$\zerosymb$s.
The states themselves are represented by \emph{state terms}~$t$
(notice that set~terms $S$ and derivate terms $D$ can both be automata and state
terms).
Intuitively, the structure of state terms records the automata constructions
used to create the top-level automaton from states of the base automata.
%
%
Non-leaf state terms, the state terms' transition function, and root state terms
are then defined inductively from base automata as described below in detail.
%
%
We will normally use $t,u$ to denote terms of all types (unless the type of the term needs to be emphasized).

\begin{example}
\label{ex:term}
Consider a~formula
$\varphi \equiv \neg\exists X.\ \singof X \wedge X = \{ \epsilon \}$
and its corresponding automata term
$
 	\termof \varphi = \Bigl\{\,\complof{
			\{
				\pi_X(\{q_0\} \conj \{p_0\})
			\}
			- \allzerotrees
 	}\,\Bigr\}
$
(we will show how $\termof \varphi$ was obtained from~$\varphi$ later).
%
For the sake of presentation,
  we will consider the following base automata for the predicates
$\sing(X)$ and  $X=\{\epsilon\}$:
$\autof{\sing(X)} = (\{q_0,q_1,q_s\},\delta,\{q_0\},\{q_1\})$ and
$\autof{ X = \{\epsilon\}} = (\{p_0,p_1,p_s\},\delta',\{p_0\},\{p_1\})$
where $\delta$ and $\delta'$ have the following sets of transitions
(transitions not defined below go to the sink states~$q_s$ and~$p_s$, respectively):
$$
\begin{array}{lcc@{\hspace{8mm}}lc}
\delta:	& (q_0,q_0)\ltr{\{X\mapsto 0\}}q_0, & (q_0,q_1)\ltr{\{X\mapsto 0\}}q_1, &
\delta': & (p_0,p_0)\ltr{\{X\mapsto 0\}}p_0, \\
			& (q_0,q_0)\ltr{\{X\mapsto 1\}}q_1, &
                        (q_1,q_0)\ltr{\{X\mapsto 0\}}q_1  &
                        &(p_0,p_0)\ltr{\{X\mapsto 1\}}p_1.
\end{array}
$$
The term $\termof\varphi$ denotes the TA
$\complof{\subsetof{(\pi_X(\autof{\sing(X)} \cap \autof{ X =
\{\epsilon\}})-\allzerotrees)}}$ constructed by the operations of
intersection, projection, derivative, subset construction, and complement.
%
\qed
\end{example}


\renewcommand{\defeq}{=}
\renewcommand{\defequiv}{\Leftrightarrow}
\begin{wrapfigure}[9]{r}{4.1cm}
\vspace*{-12.5mm}
\hspace*{-3mm}
\begin{minipage}{4.4cm}
\begin{align}
  \rt(t \disj u)        \defequiv {} & \rt(t) \lor \rt(u)
\label{eq:rtor}
\\
  \rt(t \conj u)        \defequiv {} & \rt(t) \land \rt(u)
\label{eq:rtand}
\\
  \rt(\pi_X(t))             \defequiv {} & \rt(t)
\label{eq:rtpi}
\\
  \rt(\overline t)          \defequiv {} & \neg \rt(t)
\label{eq:rtnot}
\\
  \rt(S) \defequiv {} & \exists t\in S.\,\rt(t)
\label{eq:rtset}
\\[-1mm]
  \rt(q)                    \defequiv {} & q \in \roots^\B
\label{eq:rtQ}
\end{align}
\end{minipage}
	\vspace{-2mm}
	\caption{Root term states.}
	\label{fig:root}
\end{wrapfigure}

\vspace{-4mm}
\subsubsection{Semantics of terms.}
We will define the denotation of an automata term~$\term$ as the automaton $\autof{t} = (Q,\Delta,I,R)$.
For a set automata term $\term= S$, we define
$I = S$,
$Q = \reach_\Delta(S)$
(i.e.,~$Q$ is the set of state terms reachable from the leaf state terms),
and $\Delta$ and $R$ are defined inductively to the structure of~$\term$.
Particularly,
$R$~contains the terms of $Q$ that satisfy the predicate $\rt$ defined in Fig.~\ref{fig:root},
and $\Delta$ is defined in Fig.~\ref{fig:trans},
with the addition that whenever the rules in Fig.~\ref{fig:trans} do not apply,
%
then we let $\Delta_a(t,t') = \{\emptyset\}$.
The~$\emptyset$ here is used as a~universal sink state in order to maintain~$\Delta$ complete,
which is needed for automata terms representing complements to yield the expected language.
%
\begin{wrapfigure}[11]{r}{7.2cm}
  \vspace*{-12mm}
  \hspace*{-3mm}
  \begin{minipage}{7.5cm}
    \begin{align}
      \Transof{t \disj u}{t' \disj u'} a &\defeq \Transof{t}{t'} a \timesof{\disj} \Transof{u}{u'} a
\label{eq:deltaor}
\\
      \Transof{t \conj u}{t' \conj u'} a &\defeq \Transof{t}{t'} a \timesof{\conj} \Transof{u}{u'} a
\label{eq:deltaand}
\\
      \Transof{\pi_X(t)}{\pi_X(t')} a & \defeq \{\pi_X(u)\mid u\in \Delta_{\pi_X(a)}(t, t')\}
\label{eq:deltapi}
\\
      \Transof{\overline{t}}{\overline{t'}} a & \defeq \big\{\overline{u} \mid u\in\Transof{t}{t'} a \big\}
\label{eq:deltanot}
\\
      \Transof{S}{S'} a &\defeq \bigg\{\bigcup_{t\in S, t'\in S'}\Transof{t}{t'} a\bigg\}
\label{eq:deltaset}
\\
      \Transof{q}{r} a &\defeq \delta^\B_a({q},{r})
\label{eq:deltaQ}
    \end{align}
  \end{minipage}
	\vspace{-2mm}
	\caption{Transitions among compatible state terms.}
	\label{fig:trans}
\end{wrapfigure}

The transitions of~$\Delta$ for terms of the type $\disj$, $\conj$, $\pi_X$,
$\complof{\,\cdot\,}$, and~$S$ are built from the transition function of their
sub-terms analogously to how the automata operations
of the product union, product intersection, projection, complement, and subset construction, respectively, build the transition function from the transition functions of their arguments (cf.~\secref{sec:prelims}).
The only difference is that the state terms stay \emph{annotated} with the particular operation by which they were made (the annotation of the set state terms are the set brackets).
The root states are also defined analogously as in the classical constructions.
In Figs.~\ref{fig:root} and~\ref{fig:trans},
the terms $t,t',u,u'$ are arbitrary terms, $S,S'$ are set terms, and $q,r \in Q^\B$.


Finally, we complete the definition of the term semantics by adding the
definition of semantics for the derivative term
$S - \allzerotrees$.
%
%
%
This term is a~symbolic representation of the set term that contains all state terms upward-reachable
from~$S$ in~$\autof{S}$ over~$\zerosymb$.
Formally, we first define the so-called \emph{saturation} of $\autof{S}$ as
\begin{equation}\label{eq:semric}
  \saturof {( S - \allzerotrees)} = \reach_{\Delta_{\zerosymb}}(S)
\end{equation}
%
%
(with $\reach_{\Delta_{\zerosymb}}(S)$ defined as the
fixpoint~\eqref{eq:reach}), and we complete the definition of $\Delta$ and $\rt$
in Figs.~\ref{fig:root} and~\ref{fig:trans} with three new rules to be used with
a~derivative term~$D$:
\begin{multicols}{3}
\noindent
\begin{equation}
\hspace*{-1mm}
\Transof{D}{u}{a} = \Transof{\saturof D}{u}{a}
\label{eq:deltasl}
\end{equation}
\columnbreak
\begin{equation}
\hspace*{-1mm}
\Transof{u}{D}{a} = \Transof{u}{\saturof D}{a}
\label{eq:deltasr}
\end{equation}
\columnbreak
\begin{equation}
\rt(D) \Leftrightarrow \rt(\saturof D)
\label{eq:roots}
\end{equation}
\end{multicols}
\vspace{-7mm}

\noindent
The automaton $\autof D$ then equals $\autof{\saturof{D}}$,
i.e., the semantics of a~derivative term is defined by its saturation.

\vspace{-2mm}
\begin{example}
		Let us consider a~derivative term
		$t = \{
			\pi_X(\{q_0\} \conj \{p_0\})
		\}
    - \allzerotrees$, which occurs within the nested automata term~$\termof
    \varphi$ of Example~\ref{ex:term}.
    The set term representing all terms reachable upward from $t$ is
    then the term
    \begin{align*}
      \saturof t = \{&
			\pi_X(\{q_0\} \conj \{p_0\}),
			\pi_X(\{q_1\} \conj \{p_1\}),
			\pi_X(\{q_s\} \conj \{p_s\}),\\
      &\pi_X(\{q_1\} \conj \{p_s\}),
			\pi_X(\{q_0\} \conj \{p_s\})
		 \}.
    \end{align*}
    The semantics of~$t$ is therefore the automaton $\autof t$ with the set of
    states given by~$\saturof t$.
\qed
\end{example}

\vspace{-6mm}
\subsubsection{Properties of terms.}
An implication of the definitions above, essential for termination of our
algorithm in \secref{sec:algorithm},  is that the automata represented by the
terms indeed have finitely many states.
This is the direct consequence of \Cref{lemma:finite}.
\begin{restatable}{lemma}{lemmaFinite}
\label{lemma:finite}
The size of $\reach_\Delta(t)$ is finite for any automata term $t$.
\end{restatable}

%
Intuitively, the terms are built over a finite set of states $Q^\B$, they are finitely branching,
and the transition function on terms does not increase their depth.

Let us further denote by $\seman{t}$ the language $\langof{\autof\term}$ of the automaton induced by a~term~$\term$.
\Cref{lem:sem} below shows that languages of terms can be defined from the languages of their sub-terms if the sub-terms are set terms of derivative terms.
The terms on the left-hand sides are implicit representations of the automata
operations of the respective language operators on the right-hand sides.
%
%
%
%
%
The main reason why the lemma cannot be extended to all types of sub-terms and
yield an inductive definition of term languages
is that it is not meaningful to talk about the bottom-up language of an
isolated state term that is neither a~set term nor a~derivative term (which
both are also automata terms).
%
%
%
%
This is also one of the main differences from~\cite{gaston} where every term
has its own language, which makes the reasoning and the correctness proofs in
the current paper significantly more involved.

\newcounter{savedcount}
\setcounter{savedcount}{\value{equation}}
\setcounter{equation}{0}
\renewcommand{\theequation}{{\it\alph{equation}}}

\newcommand{\T}{t}
\begin{restatable}{lemma}{lemmaSem}
  \label{lem:sem}
  For automata terms $\aterm_1, \aterm_2$ and a~set term $\sterm$, the following equalities hold:
\begin{multicols}{2}
\noindent
\begin{align}
    \label{sem:eq:deter} \semanreasonable{\{ \aterm_1\}}   &= \semanreasonable{\aterm_1} \\
    \label{sem:eq:disj} \semanreasonable{\{ \aterm_1 \disj \aterm_2 \}}     &= \semanreasonable{\aterm_1} \cup \semanreasonable{\aterm_2}\\
    \label{sem:eq:conj} \semanreasonable{\{ \aterm_1 \conj \aterm_2 \}}     &= \semanreasonable{\aterm_1} \cap \semanreasonable{\aterm_2}
\end{align}
\columnbreak
\begin{align}
    \label{sem:eq:compl} \semanreasonable{\{ \overline{\aterm_1} \}}    &= \overline{\semanreasonable{\aterm_1}}\\
    \label{sem:eq:proj} \semanreasonable{\{ \pi_X(\aterm_1) \}}        &= \pi_X(\semanreasonable{\aterm_1}) \\
    \label{sem:eq:deriv} \semanreasonable{ \sterm - \allzerotrees}  &= \semanreasonable{\sterm} - \allzerotrees
\end{align}
\end{multicols}
\end{restatable}


\setcounter{equation}{\value{savedcount}}
\renewcommand{\theequation}{\arabic{equation}}




\vspace{-10mm}
\subsubsection{Terms of formulae.}

\renewcommand{\atom}{0}
\begin{wrapfigure}[8]{r}{4.7cm}
  \vspace*{-13mm}
  \hspace*{-3mm}
  \begin{minipage}{5.0cm}
    \begin{flalign}
      \translof{\varphi_{\atom}} & {} \defeq I_{\varphi_0}\label{eq:trans0}\\
      \translof{\varphi \land \psi} & {} \defeq \translof{\varphi} \conj \translof{\psi}\\
      \translof{\varphi \lor \psi} & {} \defeq \translof{\varphi} \disj \translof{\psi}\\
      \translof{\neg \varphi} & {} \defeq \overline{\translof{\varphi}} \\[-1mm]
      \translof{\exists X.\ \varphi} & {} \defeq \{\pi_{X}(\translof{\varphi})\} - \allzerotrees\label{eq:transE}
    \end{flalign}
  \end{minipage}
	\vspace{-2mm}
	\caption{From formulae to state-terms.}
	\label{fig:transl}
\end{wrapfigure}
Our algorithm in \secref{sec:algorithm} will translate a~WS2S
formula~$\varphi$ into the automata term $\termof\varphi = \{\translof
\varphi\}$ representing a deterministic automaton with its only leaf state
represented by the state term $\translof\varphi$.
%
%
The base automata of $\termof\varphi$ include the
%
%
automaton $\autof{\varphi_\atom}$ for each atomic predicate~$\varphi_\atom$ used in~$\varphi$.
The state term $\translof\varphi$ is then defined inductively to the structure of $\varphi$ as shown in Fig.~\ref{fig:transl}.
In the definition, $\varphi_0$ is an atomic predicate, $I_{\varphi_0}$ is the set of leaf states of $\autof{\varphi_\atom}$, and $\varphi$ and $\psi$ denote
arbitrary WS2S formulae.
We note that the translation rules may create sub-terms of the form
$\{\{t\}\}$, i.e., with nested set brackets.
Since $\{\cdot\}$ semantically means determinisation by subset construction,
such double determinisation terms can be always simplified to $\{t\}$ (cf.~\Cref{lem:sem}\ref{sem:eq:deter}).
See Example~\ref{ex:term} for a~formula~$\varphi$ and its corresponding term~$\termof\varphi$.
\Cref{lem:languages_match} establishes the correctness of the formula to term translation.
\begin{restatable}{theorem}{thmLanguagesMatch} \label{lem:languages_match}
  Let $\varphi$ be a WS2S formula. Then $\langof{\varphi} = \semanreasonable{\termof\varphi}$.
\end{restatable}
%
The proof of \Cref{lem:languages_match} uses structural induction, which
is greatly simplified by \Cref{lem:sem},
but since \Cref{lem:sem} does not (and cannot, as discussed above) cover all used types of terms,
the induction step must in some cases still rely on reasoning about the definition of the transition relation on terms.


%
\section{An Efficient Decision Procedure}\label{sec:algorithm}

%
%
The development in \secref{sec:terms} already implies a na\"{\i}ve automata
term-based satisfiability check.
Namely, by \Cref{lem:languages_match}, we know that a~formula~$\varphi$ is satisfiable iff $\langofreasonable{\autof{\termof\varphi}} \neq \emptyset$.
After translating $\varphi$ into $\termof\varphi$ using rules (\ref{eq:trans0})--(\ref{eq:transE}),
we may use the definitions of the transition function and root states of $\autof{\termof\varphi} =
(Q,\Delta,I,F)$ in \secref{sec:terms} to decide the language emptiness through
evaluating the root state test $\rt(\reach_\Delta(I))$. It is enough to implement the equalities and
equivalences (\ref{eq:deltaor})--(\ref{eq:roots}) as recursive functions.
We will further refer to this algorithm as the \emph{simple recursion}.
The evaluation of $\reach_\Delta(I)$ induces nested evaluations of the fixpoint
(\ref{eq:semric}): the one on the top level of the language emptiness test and
another one for every expansion of a~derivative sub-term.
The termination of these fixpoint computations is guaranteed
due to \Cref{lemma:finite}.

Such a~na\"{\i}ve implementation is, however, inefficient and has only disadvantages in comparison to the classical decision procedure.
In this section, we will discuss how it can be optimized.
Besides an essential \emph{memoization} needed to implement the recursion efficiently,
we will show that the automata term representation is amenable to optimizations that cannot be used in the classical construction.
These are techniques of state space pruning:
the fact that the emptiness can be tested on the fly during the automata
construction allows one to avoid exploration of state space irrelevant to the test.
The pruning is done through the techniques of \emph{lazy evaluation} and \emph{subsumption}.
We will also discuss optimizations of the transition function of \secref{sec:terms} through \emph{product flattening} and \emph{nondeterministic union},
which are analogies to standard implementations of automata intersection and union.

\subsection{Memoization}\label{sec:lazy}


The simple recursion repeats the fixpoint computations that saturate derivative
terms from scratch at every call of the transition function or root test.
%
%
This is easily countered through \emph{memoization}, known, e.g., from compilers of
functional languages, which caches results of function calls
in order to avoid their re-evaluation.
Namely, after saturating a~derivative sub-term
$t = S - \allzerotrees$ of $\termof\varphi$ for the first time,
we simply \emph{replace} $t$ in $\termof\varphi$ by the saturation $\saturof t =\reach_{\Delta_{\zerosymb}}(S)$.
Since a derivative is a symbolic representation of its saturated version,
the replacement does not change the language of $\termof\varphi$.
Using memoization, every fixpoint computation is then carried out once only.


\subsection{Lazy Evaluation}

The \emph{lazy} variant of the procedure
uses \emph{short-circuiting} to optimize connectives~$\land$ and~$\lor$, and
\emph{early termination} to optimize fixpoint computation in derivative saturations.
Namely, assume that we have a~term $\termof 1 \disj \termof 2$ and that we test
whether~$\rt(\termof 1 \disj \termof 2)$.
Suppose that we establish
that $\rt(\termof 1)$; we can \emph{short circuit} the evaluation and
immediately return $\true$,
completely avoiding touching the~potentially complex term~$\termof 2$
(and analogously for a~term of the form~$\termof 1 \conj \termof 2$ when one branch
is~$\false$).
%

Furthermore, \emph{early termination} is used to optimize fixpoint computations used to saturate derivatives within tests $\rt(S - \allzerotrees)$ (obtained from sub-formulae such as~$\exists X.\ \psi$).
Namely, instead of first unfolding the whole fixpoint
into a set $\{\term_1, \ldots \term_n\}$
and only then testing whether $\rt(\term_i)$ is true for some $\term_i$,
the terms $t_i$ can be tested as soon as they are computed, and the fixpoint computation can be stopped early,
immediately when the test succeeds on one of them.
%
%
Then, instead of replacing the derivative sub-term by its full saturation,
we replace it by the partial result $\{t_1,\ldots,t_i\}-\allzerotrees$ for $i
\leq n$.
Finishing the evaluation of the fixpoint computation might later be required in order to compute a transition from the derivative.
We note that this corresponds to the concept of~\emph{continuations} from functional programming,
used to represent a~paused computation that may be required to continue later.
\begin{example}
Let us now illustrate the lazy decision procedure on our running example formula
$\varphi \equiv \neg\exists X.\ \singof X \wedge X = \{ \epsilon \}$
and the corresponding automata term
$
 	\termof\varphi = \bigl\{\,\complof{
			\{
				\projof X {\{q_0\} \conj \{p_0\}}
			\}
			- \allzerotrees
 	}\,\bigr\}
$ from Example~\ref{ex:term}.
The task of the procedure is to compute the value of $\rtof{\reach_\Delta(\termof\varphi)}$,
i.e., whether there is a root state reachable from the leaf state~$\translof\varphi$ of~$\autof{\termof\varphi}$.
The fact that $\varphi$ is ground allows us to slightly simplify the problem
because any ground formula $\psi$ is satisfiable iff $\bot \in \langof\psi$,
i.e., iff the leaf state~$\translof\psi$
of~$\autof{\termof\psi}$ is also a~root.
It is thus enough to test
$\rtof{\translof\varphi}$
where $\translof\varphi =
 	\complof{
			\{
				\projof X {\{q_0\} \conj \{p_0\}}
			\}
			- \allzerotrees
 	 }
$.

The computation proceeds as follows.
First, we use~\eqref{eq:rtnot} from Fig.~\ref{fig:root} to propagate the root test towards the derivative, i.e., to obtain that
$\rtof{\translof\varphi}$ iff $\neg\rtof{\{\projof X {\{q_0\} \conj \{p_0\}} \}
- \allzerotrees}$.
Since the $\rt$-test cannot be directly evaluated on a~derivative term, we need
to start saturating it into a~set term, evaluating~$\rt$ on the fly, hoping for
early termination.
%
We begin with evaluating the $\rt$-test on the initial element $\termof 0 =
\projof X {\{q_0\} \conj \{p_0\}}$ of the set.
The test propagates through the projection $\proj_X$ due to \eqref{eq:rtpi} and evaluates as $\false$ on the left conjunct
(through, in order, \eqref{eq:rtand}, \eqref{eq:rtset}, and \eqref{eq:rtQ})
since the state $q_0$ is not a~root state.
As a~trivial example of short circuiting, we can skip evaluating~$\rt$ on the
right conjunct~$\{p_0\}$ and conclude that $\rtof{\termof 0}$ is $\false$.

The fixpoint computation then continues with the first iteration, computing the $\zerosymb$-successors of the set~$\{\termof 0\}$.
%
We will obtain
$\Transof {\termof 0} {\termof 0} {\zerosymb} = \{\termof 0, \termof 1\}$ with
$\termof 1 = \projof X {\{q_1\} \conj \{p_1\}}$.
%
%
The test~$\rtof{\termof 1}$ now returns $\true$ because both~$q_1$ and~$p_1$ are
root states.
With that, the fixpoint computation may terminate early, with the $\rt$-test on
the derivative sub-term returning $\true$.
Memoization then replaces the derivative sub-term in $\translof\varphi$ by the partially evaluated version
$
			\{
				t_0,t_1
			\}
			- \allzerotrees
$,
and $\rtof{\translof\varphi}$ is evaluated as $\false$ due to~\eqref{eq:rtnot}.
We therefore conclude that~$\varphi$ is unsatisfiable (and invalid since it is ground).
%
%
%
\qed
\end{example}

\subsection{Subsumption}\label{sec:subsum}


The next technique we use is based on pruning out parts of a~search space that
are \emph{subsumed} by other parts.
In particular, we generalize (in a~similar way as we did for WS1S in our
previous work~\cite{gaston}) the concept used in \emph{antichain} algorithms for
efficiently deciding language inclusion and universality of finite word and tree
automata~\cite{doyen:antichain,wulf:antichains,bouajjani:antichain,abdulla:when}.
Although the problems are in general computationally infeasible (they are
$\pspace$-complete for finite word automata and $\exptime$-complete for finite
tree automata), antichain algorithms can solve them efficiently in many
practical cases.

We apply the technique by keeping set terms in the form of antichains of
\emph{simulation-maximal} elements and prune out any other simulation-smaller
elements.
Intuitively, the notion of a~term~$\term$ being simulation-smaller than $\term'$
implies that trees that might be generated from the leaf states~$T \cup
\{\term\}$ can be generated from $T \cup\{\term'\}$ too, hence discarding
$\term$ does not hurt.
Formally, we introduce the following rewriting rule:
%
\begin{align}
   \{t_1,t_2,\ldots,t_n\} \rwr \{t_2,\ldots,t_n\} \qquad \text{ for } t_1 \subsum t_2 ,
  \label{rw:subsumption}
\end{align}
which may be used to simplify set sub-terms of automata terms.
The rule~\eqref{rw:subsumption} is applied after every iteration of the
fixpoint computation on the current partial result. Hence the sequence of
partial results is monotone, which, together with the finiteness of $\reachof t
\Trans$, guarantees termination.
%
The \emph{subsumption} relation~$\subsum$ used in the rule is defined
in~Fig.~\ref{fig:subsumption} where $S \subsumfe S'$ denotes $\forall \term
\in~S\ \exists \term' \in~S'.\, \term \subsum~\term'$.
Intuitively, on base TAs, subsumption 
%
%
\begin{wrapfigure}[9]{r}{5.9cm}
\vspace*{-10mm}
\hspace*{-3mm}
\begin{minipage}{6.2cm}
\begin{align}
  \label{eq:subsum_subset}
  S &\subsum S'                          \hspace*{-4mm}& {} \defequiv {}&S \subseteq S' \lor S \subsumfe S' \\
	t \conj u &\subsum t' \conj u' \hspace*{-4mm}& {} \defequiv {}&t \subsum t' \land u \subsum u' \\
  t \disj u &\subsum t' \disj u' \hspace*{-4mm}& {} \defequiv {}&t \subsum t' \land u \subsum u' \\
  \complof{t} &\subsum \complof{t'}      \hspace*{-4mm}& {} \defequiv {}&t' \subsum t \\
	\pi_X(t) &\subsum \pi_X(t')            \hspace*{-4mm}& {} \defequiv {}&t \subsum t'
\end{align}
\end{minipage}
\vspace{-2mm}
\caption{The subsumption relation~$\subsum$}
\label{fig:subsumption}
\end{wrapfigure}
%
corresponds to inclusion of
the set terms
(the left disjunct of~\eqref{eq:subsum_subset}).
This clearly has the intended outcome: a~larger set of states can
always
simulate a~smaller set in accepting a~tree.
The rest of the definition is an inductive extension of the base case.
It can be shown that $\subsum$ for any automata term $\term$
is an upward simulation on $\autof{\term}$ in the sense of~\cite{abdulla:when}.
Consequently,
rewriting sub-terms in an automata term according to the new
rule~\eqref{rw:subsumption} does not change its language.
Moreover, the fixpoint computation interleaved with
application of rule~\eqref{rw:subsumption} terminates.

\subsection{Product Flattening}\label{sec:prod_flat}

\emph{Product flattening} is a~technique that we use to
reduce the size of fixpoint saturations that generate conjunctions and disjunctions of sets as their elements.
Consider a~term of the form $D = \{\projof X {S_0 \conj S_0'}\} - \allzerotrees$ for
a~pair of sets of terms~$S_0$ and~$S_0'$
where the TAs~$\autof{S_0}$ and~$\autof{S_0'}$ have sets of states
$Q$ and $Q'$, respectively.
The saturation generates the set $\{\projof X {S_0 \conj S'_0},\ldots,\projof X
{S_n \conj S'_n}\}$ with  $S_i\subseteq Q,S'_i\subseteq Q'$ for all $0\leq i \leq n$.
The size of this set is $2^{|Q| \cdot |Q'|}$ in the worst case.
%
%
In terms of the automata operations,
this fixpoint expansion corresponds to
first determinizing both~$\autof{S_0}$ and~$\autof{S_0'}$ and only then using the
product construction (cf.~\secref{sec:prelims}).
%
The automata intersection, however, works for nondeterministic automata
too---the determinization is not needed.
%
%
%
Implementing this standard product construction on terms would mean
transforming the original fixpoint
above into the following fixpoint with~a \emph{flattened product}:
$D = \{\projof X {S \timesof{\conj} S'}\} - \allzerotrees$ where
$\timesof{\conj}$ is the augmented product for conjunction.
This way, we can decrease the worst-case size of the fixpoint to
$|Q|\cdot|Q'|$.
A~similar reasoning holds for terms of the form $\{\projof X {S_0 \disj S_0'}\} -
\allzerotrees$.
Formally, the technique can be implemented by the following pair of
sub-term rewriting rules
where~$S$ and~$S'$ are non-empty sets of terms:
\begin{multicols}{2}
\noindent
\begin{equation}
S \disj S' \rwr S \timesof\disj S',
\label{eq:flat_disj}
\end{equation}
\columnbreak
\begin{equation}
S \conj S' \rwr S \timesof\conj S'.
\label{eq:flat_conj}
\end{equation}
\end{multicols}
\vspace{-7mm}
\noindent
Observe that for terms obtained from WS2S formulae using the translation from \secref{sec:terms},
the rules are not really helpful as is.
Consider, for instance, the term $\{\projof X {\{r\} \conj \{q\}}\} -
\allzerotrees$ obtained from a~formula $\exists X.\, \varphi \land \psi$ with
$\varphi$ and $\psi$ being atoms.
The term would be, using rule~\eqref{eq:flat_conj}, rewritten into the term
$\{\projof X {\{r \conj q\}}\} - \allzerotrees$.
Then, during a subsequent fixpoint computation, we might obtain a~fixpoint of the
following form:
$\{\projof X {\{r \conj q\}}, \projof X {\{r \conj q, r_1 \conj q_1\}},$ $\projof
X {\{r_1 \conj q_1, r_2 \conj q_2\}}\}$, where the occurrences
of the projection $\proj_X$ disallow one to perform the desired union of the inner
sets, and so the application of rule~\eqref{eq:flat_conj} did not help.
We therefore need to equip our procedure with a~rewriting rule that can be used to
push the projection inside a~set term~$S$:
\vspace{-1mm}
\begin{equation}
  \projof X S \rwr \{\projof X t \mid t \in S\}.
  \label{eq:flat_proj}
\vspace{-1mm}
\end{equation}
In the example above, we would now obtain the term $\{\projof X {r \conj q}\} -
\allzerotrees$ (we rewrote $\{\{\cdot\}\}$ to $\{\cdot\}$ as mentioned in
\secref{sec:terms}) and the fixpoint
$\{\projof X {r \conj q}, \projof X {r_1 \conj q_1}, \projof
X {r_2 \conj q_2}\}$.
The correctness of the rules is guaranteed by the following lemma:
%
\setcounter{savedcount}{\value{equation}}
\setcounter{equation}{0}
\renewcommand{\theequation}{{\it\alph{equation}}}
\begin{restatable}{lemma}{lemmaProductFlat}
  \label{lem:concise-conj}
  For sets of terms $S$ and $S'$ such that $S\neq\emptyset$ and
  $S'\neq\emptyset$, we have:
\begin{multicols}{2}
\noindent
  \begin{equation}
     \seman{\{ S \disj S' \}} = \seman{ \{S \timesof\disj S' \}},
  \end{equation}
	\begin{equation}
    \seman{\{ S \conj S' \}} = \seman{ \{S \timesof\conj S' \}},
  \end{equation}
\columnbreak
	\begin{equation}
		\seman{\{ \pi_X(S) \}} = \seman{ \{ \pi_X(t)~|~t\in S \}}.
	\end{equation}
\end{multicols}
\end{restatable}
\vspace{-7mm}

\setcounter{equation}{\value{savedcount}}
\renewcommand{\theequation}{\arabic{equation}}

However, we still have to note that there is a danger related with the
rules~\eqref{eq:flat_disj}--\eqref{eq:flat_proj}.
Namely, if they are applied to some terms in a~partially evaluated fixpoint but
not to all, the form of these terms might get different (cf. $\projof X {\{r
\conj q\}}$ and $\projof X {r \conj q}$), and it will not be possible to combine
them as source states of TA transitions when computing~$\Trans_a$, leading thus
to an incorrect result.
%
%
We resolve the situation such that we apply the rules as a pre-processing step
only before we start evaluating the top-level fixpoint, which ensures that all
terms will subsequently be generated in a compatible form.


\vspace{-1mm}
\subsection{Nondeterministic Union}\label{sec:nondetUnion}
\vspace{-1mm}

Optimization of the product term saturations from the previous section can be
pushed one step further for terms of the form $\{\projof X {S \disj S'}\} - \allzerotrees$.  
%
The idea is to use
the \emph{nondeterministic  TA union} to implement the union operation instead of the product construction.
%
The TA union is implemented as the component-wise union of the two TAs.
Its size is hence linear to the size of the input instead of quadratic as in the
case of the product (i.e., $|Q| + |Q'|$ instead of $|Q|\cdot |Q'|$).
To work correctly, the nondeterministic union requires disjoint input sets of states
(otherwise, the combination of the two transition functions could generate runs
that are not possible in either of the input~TAs).
We~implement the nondeterministic union through the following rewriting rule:
\vspace{-1mm}
\begin{equation}
  S \disj S' \rwr S \cup S' \qquad \text{for } S \not\interf S'
  \label{eq:nondet_union}
\vspace{-1mm}
\end{equation}
%
%
where~$S$ and~$S'$ are sets of terms
(similarly to \secref{sec:prod_flat}, in order to successfully reduce the
fixpoint state space on terms obtained from WS2S formulae, we also need to apply
%
%
\begin{wrapfigure}[12]{r}{7.5cm}
\vspace*{-11mm}
\hspace*{-3mm}
\begin{minipage}{7.8cm}
\begin{align}
  S \interf{}            & S'            \hspace*{-4mm}&  & {}\defequiv S = S' \lor \exists t\in S, t'\in S'.\, t\interf t' \\
  t \conj u \interf{}  & t'\conj u' \hspace*{-4mm}&  & {}\defequiv t \interf t' \lor u\interf u' \\
	t \disj u \interf{}  & t'\disj u' \hspace*{-4mm}&  & {}\defequiv t \interf t' \lor u\interf u' \\
	\overline{t} \interf{} & \overline{t'} \hspace*{-4mm}&  & {}\defequiv t \interf t' \\
  \projof X t \interf{}  & \projof X {t'}  \hspace*{-4mm}&  & {}\defequiv t \interf t' \\
  D \interf{}            & t            \hspace*{-4mm}&  & {}\defequiv \saturof{D} \interf t \\
  t \interf{}            & D            \hspace*{-4mm}&  & {}\defequiv t \interf \saturof{D} \\
  q \interf{}            & r            \hspace*{-4mm}&  & {}\defequiv  \exists 1\leq k \leq n.\, q,r \in Q_k
\end{align}
\end{minipage}
\vspace{-3mm}
\caption{Definition of interference~$\interf$}
\label{fig:interference}
\end{wrapfigure}
%
%
rule~\eqref{eq:flat_proj} to push projection inside set terms).
The relation $\interf$ used in the rule is the \emph{interference}
of terms,
defined in Fig.~\ref{fig:interference}, which generalizes the state space disjointness requirement of the
nondeterministic union of TAs.
Interference between terms tells us when we cannot perform the rewriting.
Intuitively, this happens when we obtain a~term
$\{S \disj S'\}$ where~$S$ and~$S'$ contain states from the same base
automaton $\B_k$
with the set of states~$Q_k$.

In order to avoid interference in the terms obtained from WS2S formulae,
we can perform the following pre-processing step:
When translating a~WS2S formula~$\varphi$ into a~term~$\termof \varphi$,
we create a~special version of a~base TA for every occurrence of
an atomic formula in~$\varphi$.
This way, we can never mix up terms that emerged from different subformulae to
enable a~transition that would otherwise stay disabled.

%
To use rule~\eqref{eq:nondet_union}, it is necessary to modify treatment of the
sink state $\emptyset$ in the definition of $\Delta$ of \secref{sec:terms}.
The technical difficulty we need to circumvent is that  (unlike for finite
word automata) the nondeterministic union of two (even complete) TAs is not complete.

This can cause situations such as the following:
let $D = \{\projof X {\{\complof{t}\} + \{\complof{r}\}}\} - \allzerotrees$
such that
$\Transof {t} {t} \zerosymb = \{t\}$, $\Transof {r} {r} \zerosymb = \{r\}$,
and $\rt(t)$ and $\rt(r)$ are both $\true$, i.e., both~$t$ and~$r$ can
accept any $\zerosymb$-tree, which also means that the union of their
complements should not accept any $\zerosymb$-tree.
Indeed, the saturation of~$D$ is the set term
$\saturof D = \reach_{\zerosymb}(\{\projof X {\{\complof{t}\} +
\{\complof{r}\}}\}) = \{\projof X {\{\complof{t}\} + \{\complof{r}\}}\}$
where it holds that
$\neg\rt(\projof X {\{\complof{t}\} + \{\complof{r}\}})$,
i.e., it does not accept any $\zerosymb$-tree.
On the other hand, if we use the new rule~$\eqref{eq:nondet_union}$
together with rule~\eqref{eq:flat_proj},
we obtain the term
$\{\projof X {\complof{t}}, \projof X {\complof{r}}\} - \allzerotrees$.
When computing its saturation, we will obtain a~new element
$\Transof{\projof X {\complof t}}
{\projof X {\complof r}} {\zerosymb} =
\projof X {\complof \emptyset}$.
The term $\projof X {\complof{\emptyset}}$ was constructed using
the implicit rule of \secref{sec:terms} that sends the otherwise undefined
successors of a~pair of terms to~$\{\emptyset\}$.
Note that~$\rt(\projof X {\complof \emptyset})$ is $\true$, yielding that the
fixpoint approximation
$\{\projof X {\complof t}, \projof X {\complof r}, \projof X {\complof{\emptyset}}\}$ is a~root
state, so a~$\zerosymb$-tree is accepted.
Therefore, the application of the new rule~\eqref{eq:nondet_union} changed the language.

Although the previous situation cannot happen with terms obtained from WS2S
formulae using the translation rules from \secref{sec:terms}, in order to
formulate a~correctness claim for any terms constructed using our grammar, we
remedy the issue by modifying the definition of implicit transitions of~$\Delta$
to~$\{\emptyset\}$ from \secref{sec:terms}.
Namely, the modified transition function~$\Trans_a(t_1,t_2)$ will return the
same value as before if $t_1\interf t_2$,
and otherwise it will return $\{\emptyset\}$.
We will denote the modified transition functon as~$\Delta'$ and the
corresponding semantics of a~term~$\term$ obtained using~$\Delta'$ instead
of~$\Delta$ as~$\langpof \term$
(\Cref{lem:sem,lem:concise-conj} and \Cref{lem:languages_match} could
be proved similarly with the new definition of semantics).
With these new versions of~$\Delta'$ and~$\langp$,
we can show correctness of the rule:
%
%
%
%
\begin{restatable}{lemma}{lemmaNondetUnion}
\label{lem:nondet-union}
	Let $S,S'$ be sets of terms s.t.~$S\not\interf S'$.
  Then
	$
		\langpof{\{ S \disj S' \}} = \langpof{S \cup S'}.
	$
\end{restatable}

\section{Experimental Evaluation}

We have implemented the above introduced techniques (so far with the exception
of Section~\ref{sec:nondetUnion}) in a~prototype tool written
in Haskell.\footnote{The implementation is available at
\url{https://github.com/vhavlena/lazy-wsks}.}
The base automata, hard-coded into the tool, were the TAs for the basic
predicates from \secref{sec:prelims}, together with automata for predicates
$\singof X$ and $X = \{p\}$ for a~variable~$X$ and a~fixed
tree position~$p$.
As an additional optimisation, our tool uses the so-called \emph{antiprenexing}
(proposed already in~\cite{gaston}), which pushes quantifiers down the formula
tree using the standard logical equivalences.
%
Intuitively, antiprenexing reduces the complexity of elements within fixpoints by
removing irrelevant parts outside the fixpoint.


\begin{wraptable}[11]{r}{79mm}


  \caption{Experimental results over the family of formulae
  $\varphi_n^\mathit{pt} \equiv
  \forall Z_1, Z_2.~ \exists X_1, \ldots, X_n.~ \edge(Z_1,X_1) \wedge
  \bigwedge_{i = 1}^n \edge(X_i,X_{i+1}) \wedge \edge(X_n,Z_2)$ where
  $\edge(X,Y) \equiv \edgeL(X,Y) \vee \edgeR(X,Y)$ and $\edgeLR(X,Y) \equiv
  \exists Z.\ Z = \suckleftrightof X \wedge Z \subseteq Y$.}

  \vspace*{-2mm}

  \label{tab:edge}
	\sisetup{group-minimum-digits = 4}
	\sisetup{group-separator = {,}}
	\scalebox{0.85}{
\begin{tabular}{|S[table-format=2]|S[table-format=3.3]|S[table-format=3.3]|S[table-format=3.3]|S[table-format=6.1]|S[table-format=4.1]|S[table-format=4.1]|}
	\hline

        & \multicolumn{3}{c|}{\bfseries running time (sec)} &
                \multicolumn{3}{c|}{\bfseries ~\# of subterms/states~}\\
        	\multicolumn{1}{|c|}{\textbf{$n$}} & \multicolumn{1}{c|}{\emph{Lazy}} & \multicolumn{1}{c|}{\emph{Mona}} &
                \multicolumn{1}{c|}{\emph{~Mona+AP~}} & \multicolumn{1}{c|}{\emph{Lazy}} & \multicolumn{1}{c|}{\emph{Mona}} &
                \multicolumn{1}{c|}{\emph{~Mona+AP~}} \\

        \hline

        1 &  0.02 &  0.16 &  0.15 &  149 &  216 &  216 \\
	2 &  0.50 & {\tomo} & {\tomo} &  937 & {\tomo} & {\tomo} \\
	3 &  0.83 & {\tomo} & {\tomo} &  2487 & {\tomo} & {\tomo} \\
	4 &  34.95 & {\tomo} & {\tomo}  & 8391 & {\tomo} & {\tomo} \\
	5 &  60.94 & {\tomo} & {\tomo} &  23827 & {\tomo} & {\tomo} \\

        \hline
\end{tabular}}
\end{wraptable}


We have performed experiments with our tool on various formulae and compared its
performance with that of \mona.
We applied \mona both on the original form of the considered formulae as well as
on their versions obtained by antiprenexing (which is built into our tool and
which---as we realised---can significantly help \mona too).
Our preliminary implementation of product flattening (cf.~\secref{sec:prod_flat}) is
restricted to parts below the lowest fixpoint, and our experiments showed that
it does not work well when applied on this level, where the complexity is not
too high, so we turned it off for the experiments.
We ran all experiments on a 64-bit Linux Debian workstation with the Intel(R)
Core(TM) i7-2600 CPU running at 3.40\,GHz with 16\,GiB of RAM.
We used a timeout of 100\,s.



\begin{wraptable}[9]{r}{78mm}

  \vspace*{-6.5mm}

  \caption{Experimental results over the family of formulae
  $\varphi_n^{\mathit{cnst}} \equiv
  \exists X.\ X = \{(\tleft\tright)^4\} \land X = \{(\tleft\tright)^n\}$.}

  \vspace*{-1.5mm}

  \label{tab:constants}
	\sisetup{group-minimum-digits = 4}
	\sisetup{group-separator = {,}}
	\scalebox{0.78}{
\begin{tabular}{|S[table-format=4]|S[table-format=3.3]|S[table-format=3.3]|S[table-format=3.3]|S[table-format=6.1]|S[table-format=6.1]|S[table-format=6.1]|}
	\hline

	& \multicolumn{3}{c|}{\bfseries running time (sec)} &
					\multicolumn{3}{c|}{\bfseries ~\# of subterms/states~}\\
		\multicolumn{1}{|c|}{\textbf{$n$}} & \multicolumn{1}{c|}{\emph{Lazy}} & \multicolumn{1}{c|}{\emph{Mona}} &
					\multicolumn{1}{c|}{\emph{~Mona+AP~}} & \multicolumn{1}{c|}{\emph{Lazy}} & \multicolumn{1}{c|}{\emph{Mona}} &
					\multicolumn{1}{c|}{\emph{~Mona+AP~}} \\

        \hline

	80 &  14.60 &  40.07 &  40.05 &  1146 &  27913 &  27913 \\
	90 &  21.03 &  64.26 &  64.20 &  1286 &  32308 &  32308 \\
	100 &  28.57 &  98.42 &  98.91 &  1426 &  36258 &  36258 \\
	110 &  38.10 & {\tomo} & {\tomo} &  1566 & {\tomo} & {\tomo} \\
	120 &  49.82 & {\tomo} & {\tomo} &  1706 & {\tomo} & {\tomo} \\

        \hline
\end{tabular}}
\end{wraptable}


We first considered various WS2S formulae on which \mona was successfully
applied previously in the literature.
%
On them, our tool is quite slower than \mona, which is not much
surprising given the amount of optimisations built into \mona
(for instance, for the benchmarks from~\cite{strand1}, \mona on average
took~0.1\,s, while we timeouted).
Next, we identified several parametric families of formulae (adapted
from~\cite{gaston}), such as, e.g.,
$\varphi_n^{\mathit{horn}} \equiv \exists X.~ \forall X_1.~ \exists X_2, \ldots X_n.~ ( (X_1
\subseteq X \wedge X_1 \neq X_2) \Rightarrow X_2 \subseteq X) \wedge \ldots
\wedge ( (X_{n-1} \subseteq X \wedge X_{n-1} \neq X_n) \Rightarrow X_n \subseteq
X)$,
where our approach finished within 10\,ms, while the time of \mona was
increasing when increasing the parameter~$n$, going up to 32\,s for $n = 14$ and
timeouting for $k \geq 15$.
%
It turned out that \mona could, however, easily handle these formulae after
antiprenexing, again (slightly) outperforming our tool.
Finally, we also identified several parametric families of formulae that \mona
could handle only very badly or not at all, even with antiprenexing, while our tool can handle
them much better.
These formulae are mentioned in the captions of Tables~\ref{tab:edge},
\ref{tab:constants}, and \ref{tab:left-right}, which give detailed results of the
experiments.


\begin{wraptable}[8]{r}{72mm}

        \vspace*{-8mm}

        \caption{Experiments over the family $\varphi_n^{\mathit{sub}} = \forall
        X_1,\dots,X_n$ $\exists X.\ \bigwedge_{i = 1}^{n-1}X_i\subseteq X
        \Rightarrow (X_{i+1} = \suckleftof{X} \vee X_{i+1} = \suckrightof{X}).$}

  \vspace*{-2mm}

  \label{tab:left-right}
	\sisetup{group-minimum-digits = 4}
	\sisetup{group-separator = {,}}
	\scalebox{0.80}{
\begin{tabular}{|S[table-format=2]|S[table-format=2.3]|S[table-format=3.3]|S[table-format=3.3]|S[table-format=5.1]|S[table-format=4.1]|S[table-format=4.1]|}

	\hline

	& \multicolumn{3}{c|}{\bfseries running time (sec)} &
					\multicolumn{3}{c|}{\bfseries ~\# of subterms/states~}\\
		\multicolumn{1}{|c|}{\textbf{$n$}} & \multicolumn{1}{c|}{\emph{Lazy}} & \multicolumn{1}{c|}{\emph{Mona}} &
					\multicolumn{1}{c|}{\emph{~Mona+AP~}} & \multicolumn{1}{c|}{\emph{Lazy}} & \multicolumn{1}{c|}{\emph{Mona}} &
					\multicolumn{1}{c|}{\emph{~Mona+AP~}} \\

        \hline

        3 & 0.01 & 0.00 & 0.00 & 140 & 92 & 92 \\
        4 & 0.04 & 34.39 & 34.47 & 386 & 170 & 170 \\
        5 & 0.24 & {\tomo} & {\tomo} & 981 & {\tomo} & {\tomo} \\
        6 & 2.01 & {\tomo} & {\tomo} & 2376 & {\tomo} & {\tomo} \\

        \hline
\end{tabular}}
\end{wraptable}


Particularly, Columns 2--4 give the running times (in seconds) of our tool
(denoted \emph{Lazy}), \mona, and \mona
with
antipre\-nexing.
Columns 5--7 characterize the size of the generated terms and automata.
Namely, for our approach, we give the number of nodes in the final term tree
(with the leaves being states of the base TAs).
For \mona, we give the sum of the numbers of states of all the minimal
deterministic TAs constructed by \mona when evaluating the formula.
The ``--'' sign means a timeout or memory shortage.

The formulae considered in Tables~\ref{tab:edge}--\ref{tab:left-right} speak
about various paths in trees.
We were originally inspired by formulae kindly provided by Josh
Berdine, which arose from attempts to translate separation logic formulae to
WS2S (and use \mona to discharge them), which are beyond the capabilities of \mona (even with
antiprenexing).
We were also unable to handle them with our tool, but
%
%
our experimental results on the tree path formulae indicate (despite the prototypical
implementation) that our techniques can help one to handle some
complex graph formulae that are out of the capabilities of \mona.
Thus, they provide
a~new line of attack on deciding hard WS2S formulae, complementary to the
heuristics used in \mona.
Improving the techniques and combining them with the classical approach of
\mona is a~challenging subject for our future work.

\section{Related Work} \label{sec:related}

The seminal works~\cite{Buchi60,Rabin69} on the automata-logic connection were the milestones leading to what we call here the classical tree automata-based decision procedure for \wsks~\cite{Thatcher68}.
Its non-elementary worst-case complexity was proved in~\cite{Stockmeyer73},
and the work~\cite{glenn-wia-96} presents the first implementation, restricted to WS1S,
with the ambition to use heuristics to counter the high complexity.
The authors of \cite{tata}~provide an excellent survey of the classical results and literature related to \wsks and tree automata.

The tool \mona{}~\cite{monapaper} implements the classical decision procedures
for both WS1S and WS2S.
It is still the standard tool of choice
for deciding WS1S/\wsks formulae due to its all-around most robust performance.  The efficiency of \mona{} stems from many
optimizations, both higher-level (such as automata minimization, the encoding of
first-order variables used in models, or the use of multi-terminal BDDs to encode the
transition function of the automaton) as well as lower-level (e.g.~optimizations
of hash tables, etc.)~\cite{monasecrets,monarestrictions}.
The M2L(Str) logic, a dialect of WS1S, can also be decided by a~similar automata-based
decision procedure, implemented within, e.g., \osel{}~\cite{jmosel} or the
symbolic finite automata framework of~\cite{veanes}.
In particular, \osel{}
implements several optimizations (such as second-order value
numbering~\cite{osel-numbering}) that allow it to outperform \mona{} on some
benchmarks (\mona{} also provides an~M2L(Str) interface on top of the WS1S
decision procedure).

The original inspiration for our work are the antichain techniques for checking
universality and inclusion of finite
automata~\cite{doyen:antichain,wulf:antichains,bouajjani:antichain,abdulla:when}
and language emptiness of alternating automata \cite{afaantichain}, which use
symbolic computation together with subsumption to prune large state spaces
arising from subset construction.
This paper is a continuation of our work on WS1S, which started by
\cite{fiedor:tacas15}, where we discussed a basic idea of generalizing the
antichain techniques to a~WS1S decision procedure.
In \cite{gaston}, we then presented a~complete WS1S decision procedure based on
these ideas that is capable to rival \mona{} on already interesting benchmarks.
The work in~\cite{traytel:coalgebras} presents a decision procedure that,
although phrased differently, is in essence fairly similar to that
of~\cite{gaston}.
%
This paper generalizes~\cite{gaston} to WS2S.  It is not merely a
straightforward generalization of the word concepts to trees.  A nontrivial
transition was needed from language terms of~\cite{gaston}, with their semantics
being defined straightforwardly from the semantics of sub-terms, to tree
automata terms, with the semantics defined as a language of an automaton with
transitions defined inductively to the structure of the term.  This change makes
the reasoning and correctness proof considerably more complex, though the
algorithm itself stays technically quite simple.

Finally, Ganzow and Kaiser~\cite{ganzow:new} developed a~new decision procedure
for the weak mon\-a\-dic second-order logic on inductive structures within
their tool \toss{}.
Their approach completely avoids automata; instead, it is based on the Shelah's
composition method.
The paper reports that the \toss{} tool could outperform \mona{} on two families
of WS1S formulae, one derived from Presburger arithmetics and one formula of the
form that we mention in our experiments as problematic for \mona{} but solvable
easily by \mona{} with antiprenexing.

\subsubsection*{Acknowledgement}

We thank the anonymous reviewers for their helpful comments on how to improve
the exposition in this paper.
This work was supported by
the Czech Science Foundation project 17-12465S,
the FIT BUT internal project FIT-S-17-4014,
and The Ministry of Education, Youth and Sports from the
National Programme of Sustainability (NPU~II) project IT4Innovations
excellence in science---LQ1602.




\bibliographystyle{splncs}
\bibliography{bibliography,citaceMONA}



\vfill
\eject

\appendix


\vspace{-0.0mm}
\section{Proofs}\label{app:proofs}
\vspace{-0.0mm}

In the proofs we use an alternative definition of automata term semantics. First
we bring a notion of a \emph{term expansion} and an \emph{expanded term}.
Expanded term does not contain a derivative term as a subterm. Term expansion is
then defined recursively as follows:
\begin{inparaenum}[(i)]
	\item $\expt{t} = t$ if $t$ is expanded.
	\item $\expt{t} = \expt{(t[u/\saturof{u}])}$ where $u$ is a derivative term of
		the form $S-\alltreesx\Gamma$ where $S$ is a expanded term.
\end{inparaenum}
Intuitively in the term expansion, derivative subterms are saturated in a
bottom-up manner. Then, we have $\langof{\autof{\expt{t}}} = \langof{\autof{t}}$
and therefore, $\langof{\expt{t}} = \langof{t}$.


\lemmaFinite*

\begin{proof}
(Sketch) First, we define \emph{depth} of a term $t$ inductively as follows:
\begin{inparaenum}[(i)]
	\item $\theiof{q} = 1$ for $q\in Q^\B$,
	\item $\theiof{t_1 \circ t_2} = 1 + \max(\theiof{t_1}, \theiof{t_2})$ for
		$\circ\in\{ \conj, \disj \}$,
	\item $\theiof{\diamond t_1} = 1 + \theiof{t_1}$ for $\diamond\in\{ \pi_X,
		\overline{\cdot} \}$,
	\item $\theiof{S} = 1 + \max_{t\in S}(\theiof{t})$, and
	\item $\theiof{S - \alltreesx\syms} = 1 + \theiof{S}$.
\end{inparaenum}
Then since the number of reachable states in base automata is finite, for a
given $n$ there is a finite number of terms of depth at most $n$. Moreover, for
two terms $t_1$ and $t_2$ and each $ t\in\Transof{t_1}{t_2} a$ we have
$\theiof{t} \leq \max(\theiof{t_1}, \theiof{t_2})$. Therefore, for an automaton
term $S$ it holds that $\reach_\Delta(S)$ is finite.

\qed
\end{proof}

\medskip


\lemmaSem*

\begin{proof}
	\eqref{sem:eq:deter}: We prove more general form of \eqref{sem:eq:deter}
	namely $\seman{\{ \aterm_1, \dots, \aterm_n \}} = \seman{\bigcup_{1\leq i\leq
	n}\expt{\aterm_i}}$ $(\subseteq)$ We start with the following reasoning:
	$\tau\in\seman{\{ \aterm_1, \dots, \aterm_n \}}$ iff there is accepting run
	$\rho$ on $\tau$ in $\autof{\{ \expt{\aterm_1}, \dots, \expt{\aterm_n} \}}$
	having all leaf states from $\{ \expt{\aterm_1}, \dots, \expt{\aterm_n} \}$.
	For simplicity we set $\Xi = \reachxof {\bigcup_{1\leq i\leq
	n}\expt{\aterm_i}} \Trans \Sigma$. Moreover, $\forall
	w\in\domof{\tau}\setminus\leafsof \tau,\ t\in\rho(w)$ we have that $\exists
	t_1 \in\rho(w.\tleft), t_2 \in\rho(w.\tright):\
	t\in\Transof{t_1}{t_2}{\tau(w)} \subseteq \rho(w)$. Since this run is
	accepting, there is a $r\in\rho(\epsilon)$ s.t. $\rt(r)$. Therefore, we are
	able to construct the mapping $\rho'$ on $\domof{\tau}$ defined as
	$\rho'(\epsilon) = r$, $\rho'(w)
	\in\Transof{\rho'(w.\tleft)}{\rho'(w.\tright)}{\tau(w)}$, and $\rho'(w) \in
	\rho(w)$ for $w\in\domof \tau$. Hence $\forall w \in\leafsof \tau:\
	\rho'(w)\in\bigcup_{1\leq i\leq n}\expt{\aterm_i}$. It means that $\rho'(w)
	\in \Xi$ for each $w\in\domof t$, and therefore $\rho'$ is an accepting run on
	$\tau$ in $\autof{\bigcup \expt{\aterm_i}}$, i.e., $\tau\in
	\seman{\bigcup_{1\leq i\leq n}\aterm_i}$.

	$(\supseteq)$ Consider $\tau\in \seman{\bigcup_{1\leq i\leq
	n}\expt{\aterm_i}}$. Then there is an accepting run $\rho$ on $\tau$ in
	$\autof{\bigcup \expt{\aterm_i}}$. We can then construct the mapping $\rho'$
	on $\domof{\tau}$ defined as $\rho'(u) = S_u$ and $\rho'(w)
	\in\Transof{\rho'(w.\tleft)}{\rho'(w.\tright)}{\tau(w)}$ for $u\in\leafsof
	\tau, w\in\domof \tau$ where $\rho(u) \in S_u \wedge S_u = \aterm_i$ for some
	$1\leq i \leq n$. We have that $\forall w \in\domof \tau:\ \rho(w) \in
	\rho'(w)$ and therefore $\rho'$ is an accepting run on $\tau$ in $\autof{\{
	\expt{\aterm_1}, \dots, \expt{\aterm_n} \}}$, i.e., $\tau\in\seman{\{
	\aterm_1, \dots, \aterm_n\}}$.
	\bigskip

	\eqref{sem:eq:disj}: $(\subseteq)$ We again start with the reasoning:
	$\tau\in\seman{\{ \aterm_1 \disj \aterm_2 \}}$ iff there is accepting run
	$\rho$ on $\tau$ in $\autof{\{ \expt{\aterm_1} \disj \expt{\aterm_2} \}}$.
	Further since $\rho$ is accepting, we can define mappings $\rho_1$, $\rho_2$
	on $\domof\tau$ s.t. $\forall w\in\domof\tau:\ \rho_1(w) = l(\rho(w)) \wedge
	\rho_1(w) = r(\rho(w))$ where $l(S_1\disj S_2) = S_1$, $r(S_1\disj S_2) =
	S_2$. Moreover, $\rho_1$ is a run on $\tau$ in $\autof{\{ \expt{\aterm_1} \}}$
	and $\rho_2$ is a run in $\autof{\{\expt{\aterm_2} \}}$. We also have
	$\rt(\rho(\epsilon))$ hence $\rt(\rho_1(\epsilon)) \vee
	\rt(\rho_2(\epsilon))$. Therefore $\tau\in\langof{\autof{\{ \expt{\aterm_1}
	\}}} \vee \tau\in \langof{\autof{\{ \expt{\aterm_2} \}}}$, i.e.,
	$\tau\in\seman{\{ \aterm_1 \}} \cup \seman{\{ \aterm_2 \}}$ and from
	\eqref{sem:eq:deter} we get the desired form.

	$(\supseteq)$ Consider $\tau\in\seman{ \aterm_1 } \cup \seman{ \aterm_2 }$.
	From \eqref{sem:eq:deter} we get $\tau\in\seman{\{ \aterm_1 \}} \cup \seman{\{
	\aterm_2 \}}$. Then there are runs $\rho_1$ in $\autof{\{ \expt{\aterm_1} \}}$
	and $\rho_2$ in $\autof{\{ \expt{\aterm_2} \}}$ on $\tau$ s.t. at least one of
	them is accepting. We can define mapping $\rho$ on $\domof\tau$ s.t. $\forall
	w\in\domof\tau:\ \rho(w) = \rho_1(w) \disj \rho_2(w)$. Such defined mapping is
	an accepting run on $\tau$ in $\autof{\{ \expt{\aterm_1} \disj \expt{\aterm_2}
	\}}$. Therefore $\tau\in\seman{\{ \aterm_1 \disj \aterm_2 \}}$. \bigskip


	\eqref{sem:eq:conj}: Analogy to $(b)$.
	\bigskip

	\eqref{sem:eq:compl}: We start with the following reasoning: $\tau\in\seman{\{
	\overline{\aterm_1} \}}$ iff there is accepting run $\rho$ on $\tau$ in
	$\autof{\{ \overline{\expt{\aterm_1}} \}}$. Since in $\autof{\{
	\overline{\expt{\aterm_1}} \}}$ there is only one leaf state and for each
	$a\in\Sigma$: $|\Transof{\overline{S_1}}{\overline{S_2}} a|\leq 1$, there is
	at most one accepting run on each tree. The same holds also for $\autof{\{
	\expt{\aterm_1} \}}$. Note that both $\autof{\{ \expt{\aterm_1} \}}$ and
	$\autof{\{ \overline{\expt{\aterm_1}} \}}$ are complete. Therefore $\rho$ is a
	run on $\tau$ in $\autof{\{ \expt{\aterm_1} \}}$ iff $\overline\rho$ is a run
	on $\tau$ in $\autof{\{ \overline{\expt{\aterm_1}} \}}$ where $\forall
	w\in\domof\tau:\ \overline\rho(w) = \overline{\rho(w)}$. From the definition
	of $\rt$ we further have $\neg\rt(\rho(\epsilon)) \Leftrightarrow
	\rt(\overline\rho(\epsilon))$. Therefore $\rho$ is not accepting in $\autof{\{
	\expt{\aterm_1} \}}$ iff $\overline\rho$ is accepting in $\autof{\{
	\overline{\expt{\aterm_1}} \}}$, which implies $\tau\in L(\autof{\{
	\overline{\expt{\aterm_1}} \}})$ iff $\tau\notin L(\autof{\{ \expt{\aterm_1}
	\}})$ and from \eqref{sem:eq:deter} we get the desired form.
	\bigskip

	\eqref{sem:eq:proj}: $(\subseteq)$ Consider $\tau\in\seman{\{ \pi_X(\aterm_1)
	\}}$. Then there is an accepting run $\rho$ on $\tau$ in $\autof{\{
	\pi_X(\expt{\aterm_1}) \}}$. From the definition of transition function we get
	that there is the accepting run $\rho'$ on some $\tau'$ in $\autof{\{
	\expt{\aterm_1} \}}$ where $\tau\in\pi_X(\tau')$ and $\forall w\in\domof\tau:\
	\rho(w) = \pi_X(\rho'(w))$. Therefore, $\tau\in \pi_X(\seman{\{\aterm_1\}}) =
	\pi_X(\seman{\aterm_1})$.



	$(\supseteq)$ Consider $\tau\in\pi_X(\seman{\aterm_1})$. Then, there is $\tau'
	\in\seman{\aterm_1}$ s.t. $\tau\in\pi_X(\tau')$. According to the part
	\eqref{sem:eq:deter}, there is an accepting run $\rho$ on $\tau'$ in
	$\autof{\{ \expt{\aterm_1} \}}$. Then there is also the accepting run $\rho'$
	on $\tau$ in $\autof{\{ \pi_X(\expt{\aterm_1}) \}}$ where $\forall
	w\in\domof\tau:\ \rho'(w) = \pi_X(\rho(w))$, which concludes the proof.
	\bigskip

  \eqref{sem:eq:deriv}: We prove more general form of the equality,
  $\seman{\aterm_1}-\alltreesx\syms = \seman{\aterm_1-\alltreesx\syms}$ for a
  set of symbols $\syms$. Note that $\aterm_1$ is a set term. In the following
  text, for a set term $S$ and a set of symbols $\syms$ we define $S
  \ominus\Gamma = \expt{S} \cup \bigcup\{\Transof {t_1} {t_2} a \mid t_1, t_2\in
  \expt{S}, a\in \Gamma\}$. Note since $\syms\subseteq\Sigma$, we have
  $\reachxof{\expt{S}}{\Trans}{\Sigma} = \reachxof{S
  \ominus\Gamma}{\Trans}{\Sigma}$. Moreover, a set of trees of height at most
  $n$ containing symbols from $\syms$ we denote by $\syms^n$. Formally, $\syms^n
  = \{ t\in\alltreesx\syms~|~ \forall w\in \domof{t}:\ |w|\leq n \}$. Note that
  $|w|$ denotes the length of a word $w$. We begin with a claim $\seman{S
  \ominus\syms} = \seman{S} - \syms^1$.

  $\subseteq$: Consider a tree $\tau\in\seman{S \ominus\syms}$. Therefore there
  is an accepting run $\rho$ on $\tau$ in $\autof{S \ominus\syms}$ having leaf
  states in $S \ominus\syms$. Moreover, for each $w\in\leafsof\tau$ s.t.
  $\rho(w)\notin\expt{S}$ it holds that $\exists t_\tleft^w,
  t_\tright^w\in\expt{S}, a\in\syms:\ \rho(w)
  \in\Transof{t_\tleft^w}{t_\tright^w} a$. Hence, we can extend the run $\rho$
  to $\rho'$ defined as $\rho'_{|\domof\tau} = \rho$ and $\forall
  w\in\leafsof\tau, \rho(w)\notin\expt{S}:\ \rho'(w.\tleft) = t_\tleft^w \wedge
  \rho'(w.\tright) = t_\tright^w$. The mapping $\rho'$ is a run in
  $\autof{\expt{S}}$ on a tree $\tau'\in\seman{S}$ where $\tau\in \tau' -
  \syms^1$, and hence $\tau\in \seman{S} - \syms^1$.

	$\supseteq$: Consider $\tau\in\seman{S} - \syms^1$. Then there is a
	$\tau'\in\seman{S}$ s.t. $\tau \in\tau'-\syms^1$. Hence there is an accepting
	run $\rho'$ on $\tau'$ in $\autof{\expt{S}}$. Now consider the set $\Theta =
	\{ w\in\leafsof\tau~|~\rho'(w)\notin \expt{S}\}$.  Since $\tau
	\in\tau'-\syms^1$, we have $\forall w\in\Theta:\ \rho'(w.\tleft)
	\in\expt{S}\wedge\rho'(w.\tright) \in\expt{S}\wedge \tau'(w)\in\syms$.
	Therefore, $\rho = \rho'_{|\domof\tau}$ is an accepting run on $\tau$ in
	$\autof{S \ominus\syms}$, i.e., $\tau\in\seman{S \ominus\syms}$.


  We proceed to main part of the lemma. Consider a sequence of automata terms
  $S_0 = \expt{\sterm}$, $S_1 = S_0 \ominus\syms$, $S_{i+1} = S_i \ominus\syms$.
  Because the set of all terms that can occur in $S_i$ is finite
  (Lemma~\ref{lemma:finite}), there is some $n_0$ s.t. for all $n'' \geq n_0$
  and $ n' \geq n_0$ we have $S_{n'} = S_{n''}$. Moreover, $S_{n_0} =
  \reachxof{\expt{\sterm}}{\Trans}{\syms}$. From the previous claim we have
  $\seman{S_i} = \seman{\sterm}-\syms^i$ and consequently $\bigcup_{i \geq
  1}\seman{S_i} = \bigcup_{i \geq 1}\seman{\sterm}-\syms^i = \seman{\sterm} -
  \alltreesx\syms$. Moreover from the previous reasoning we have $\seman{S_{n'}}
  = \seman{S_{n''}}$ for $n'' \geq n_0, n' \geq n_0$. Hence $\bigcup_{i \geq
  1}\seman{S_i} = \seman{S_{n_0}}$ (follows from $\seman{S_i} \subseteq
  \seman{S_{i+1}}$). Finally we have $\seman{S_{n_0}} =
  \seman{\reachxof{\expt{\sterm}}{\Trans}{\syms}} =
  \seman{\sterm-\alltreesx\syms} = \seman{\sterm} - \alltreesx\syms$.
	\qed
\end{proof}

\medskip


\thmLanguagesMatch*

\begin{proof}
For the purpose of this proof we restrict the definition of terms to
\emph{deterministic terms} constructed using the following grammar:
\begin{align}
	D &::= \{ d,\dots,d \}~|~\{ \pi_X(d),\dots,\pi_X(d) \} \\
  d &::= S~|~d \disj d~|~d \conj d~|~\overline d~|~D~|~D - \alltreesx\syms
\end{align}
where $D$ is a finite set of deterministic terms and $S$ is a finite set of
terms. Note that for two expanded deterministic terms $t_1$, $t_2$ we have
$|\Transof{t_1}{t_2} a| = 1$. Further note that for a WS2S formula $\varphi$,
$\translof{\varphi}$ is a deterministic term.

Now, we prove $\langof\varphi = \seman{\{\translof{\varphi}\}}$ by a structural
induction on $\varphi$. We use properties of the classical decision procedure.
\begin{itemize}
  \item[--] $\varphi = \varphi_{\atom}$ where $\varphi_{\atom}$ is an atomic
		formula: From the translation formula to terms and
		Lemma~\ref{lem:sem}~\eqref{sem:eq:deter} we directly have
		$\langof{\varphi_{\atom}} = \seman{\translof{\varphi_{\atom}}} =
		\seman{\{\translof{\varphi_{\atom}}\}} $.
  \item[--] $\varphi = \psi_1 \land \psi_2$:
    From the translation formula to terms and
		Lemma~\ref{lem:sem}~\eqref{sem:eq:deter} we get
    \begin{equation}\label{eq:corr-conj}
      \begin{split}
      \seman{\{\translof{\varphi}\}} =
      \seman{\{ \translof{\psi_1} \conj \translof{\psi_2} \}} =
      \seman{\{\{ \translof{\psi_1} \conj \translof{\psi_2} \}\}}.
      \end{split}
    \end{equation}
    Further, from \eqref{eq:corr-conj}, Lemma~\ref{lem:concise-conj} and
    Lemma~\ref{lem:sem}~\eqref{sem:eq:conj} we obtain
    \begin{equation}
      \seman{\{\{ \translof{\psi_1} \conj \translof{\psi_2} \}\}} = \seman{\{
      \{\translof{\psi_1}\} \conj \{\translof{\psi_2} \}\}} =
      \seman{\{\translof{\psi_1}\}} \cap \seman{\{\translof{\psi_2}\}}.
    \end{equation}
    Finally from IH we have $\seman{\{\translof{\varphi}\}} = \langof{\psi_1} \cap
		\langof{\psi_2} = \langof\varphi$.

  \item[--] $\varphi = \psi_1 \lor \psi_2$:
    From the translation formula to terms and
		Lemma~\ref{lem:sem}~\eqref{sem:eq:deter} we get
    \begin{equation}\label{eq:corr-disj}
      \begin{split}
      \seman{\{\translof{\varphi}\}} &=
      \seman{\{ \translof{\psi_1} \disj \translof{\psi_2} \}} =
      \seman{\{\{ \translof{\psi_1} \disj \translof{\psi_2} \}\}}.
      \end{split}
    \end{equation}
    Further, from \eqref{eq:corr-disj}, Lemma~\ref{lem:concise-conj} and
    Lemma~\ref{lem:sem}~\eqref{sem:eq:disj} we obtain
    \begin{equation}
      \seman{\{ \{\translof{\psi_1}\} \disj \{\translof{\psi_2} \}\}} =
      \seman{\{\translof{\psi_1}\}} \cup \seman{\{\translof{\psi_2}\}}.
    \end{equation}
    Finally from IH we have $\seman{\{\translof{\varphi}\}} = \langof{\psi_1}
    \cup \langof{\psi_2} = \langof\varphi$.
  \item[--] $\varphi = \neg\psi$:
		First, we prove the following claim: Let $t$ be a deterministic term, then
		$\seman{\left\{\overline{\{t\}}\right\}} =
		\seman{\left\{\overline{t}\right\}}$. Proof: First consider two expanded
		deterministic terms $t_1, t_2$. Since $t_1, t_2$ are deterministic, we have
		$\Transof{t_1}{t_2} a = \{ t' \}$ for some deterministic term $t'$.
		Therefore, $\Transof{\overline{t_1}}{\overline{t_2}} a = \{ \overline{t'}
		\}$ and $\Transof{\overline{\{t_1\}}}{\overline{\{t_2\}}} a = \{
		\overline{\{t'\}} \}$. Hence, there is an accepting run $\rho$ on a tree
		$\tau$ in $\autof{\left\{\overline{\{t\}}\right\}}$ iff there is an
		accepting run $\rho'$ on a tree $\tau$ in
		$\autof{\left\{\overline{t}\right\}}$ where $\forall w\in\domof\tau:\
		\rho(w) = \overline{s} \wedge \rho'(w) = \overline{\{s\}}$.

    We proceed in the main part of the theorem. From the translation formula to
    terms and from the previous claim we get
    \begin{equation}\label{eq:corr-neg}
      \seman{\{\translof{\varphi}\}} =
      \seman{\left\{\overline{\translof{\psi}}\right\}} =
      \seman{\left\{\overline{\{\translof{\psi}\}}\right\}}.
    \end{equation}
    Finally from \eqref{eq:corr-neg}, Lemma~\ref{lem:sem}~\eqref{sem:eq:compl}
		and IH we have
    \begin{equation}
      \seman{\{\translof{\varphi}\}} =
      \overline{\seman{\{ \translof{\psi} \} }} =
      \overline{\langof\psi} = \langof{\varphi}.
    \end{equation}
  \item[--] $\varphi = \exists X.\ \psi$:
		First, we prove the following claim: Let $t$ be a deterministic term, then
		$\seman{\left\{\pi_X(\{t\})\right\}} = \seman{\left\{\pi_X(t)\right\}}$.
		Proof: First consider two expanded deterministic terms $t_1, t_2$. Since
		$t_1, t_2$ are deterministic, for each $a$ we have $\Transof{t_1}{t_2} a =
		\{ t_a \}$ for some deterministic term $t_a$. Therefore,
		$\Transof{\pi_X(t_1)}{\pi_X(t_2)} a = \{ \pi_X(t_b)~|~b\in\pi_X(a) \}$ and
		$\Transof{\pi_X(\{t_1\})}{\pi_X(\{t_2\})} a = \{
		\pi_X(\{t_b\})~|~b\in\pi_X(a) \}$. Hence, there is an accepting run $\rho$
		on a tree $\tau$ in $\autof{\left\{\pi_X(\{t\})\right\}}$ iff there is an
		accepting run $\rho'$ on a tree $\tau$ in $\autof{\left\{\pi_X(t)\right\}}$
		where $\forall w\in\domof\tau:\ \rho(w) = \pi_X(s) \wedge \rho'(w) =
		\pi_X(\{s\})$.

    We proceed in the main part of the theorem. From the translation formula to
		terms we get
		\begin{align}\label{eq:corr-proj}
      \seman{\{\translof{\varphi}\}} =
      \seman{\left\{\pi_{X}(\translof{\psi})\right\} - \alltreesx{\zerosymb} }
    \end{align}
    Further, from \eqref{eq:corr-proj},
    Lemma~\ref{lem:sem}~\eqref{sem:eq:deriv}, and the previous claim we have
    \begin{equation}\label{eq:corr-minus}
      \seman{\{\translof{\varphi}\}} =
      \seman{\left\{\pi_{X}(\translof{\psi})\right\}} - \alltreesx{\zerosymb} =
			\seman{\left\{\pi_{X}(\{\translof{\psi}\})\right\}} - \alltreesx{\zerosymb}.
    \end{equation}
    Then from \eqref{eq:corr-minus} and Lemma~\ref{lem:sem}~\eqref{sem:eq:proj}
		we obtain
    \begin{equation}\label{eq:corr-fin}
      \seman{\{\translof{\varphi}\}} =
			\pi_X\left(\seman{\left\{\translof{\psi}\right\}}\right) - \alltreesx{\zerosymb}.
    \end{equation}
    IH together with~\eqref{eq:corr-fin} give us
    \begin{equation}
      \seman{\{\translof{\varphi}\}} =
      \pi_{X}(\langof\psi) - \alltreesx{\zerosymb} = \langof\varphi.
    \end{equation}
\end{itemize}
Finally, we have $L(\varphi) =
\seman{\{ \translof{\varphi} \}}$.
\qed
\end{proof}

\medskip

%
%
%


\lemmaProductFlat*

\begin{proof}
$(a)$: $(\subseteq)$: Consider some $\tau\in\seman{\{S \disj S'\}}$. From
Lemma~\ref{lem:sem} we have $\seman{\{S \disj S'\}} = \seman{S} \cup
\seman{S'}$. Hence there are runs $\rho_1$ in $\autof{\expt{S}}$ and $\rho_2$ in
$\autof{\expt{S'}}$ on $\tau$ and at least one them is accepting (both runs
exist since the transition function $\Delta$ is total). Then, we can construct a
mapping $\rho$ on $\domof\tau$ defined as $\forall w\in\domof\tau:\ \rho(w) =
\rho_1(w)+\rho_2(w)$. The $\rho$ is a run on $\tau$ in $\autof{\{ \expt{t_1}
\disj \expt{t_2}~|~ t_1\in S, t_2\in S' \}}$. Moreover, this run is accepting
since $\rho_1$ or $\rho_2$ is accepting. Therefore, $\tau\in\seman{\{ t_1 \disj
t_2~|~ t_1\in S, t_2\in S' \}}$ and from Lemma~\ref{lem:sem} $\tau\in\seman{\{S
\timesof\conj S' \}}$.

$(\supseteq)$: Consider some $\tau\in\seman{\{S \timesof\conj S' \}}$. Then from
Lemma~\ref{lem:sem} we obtain that $\tau\in\\\seman{\{ t_1 \disj t_2~|~t_1\in S,
t_2\in S' \}}$. Then, there is the accepting run $\rho$ on $\tau$ in $\autof{\{
\expt{t_1} \disj \expt{t_2}~|~ t_1\in S, t_2\in S' \}}$. Further, we are able to
construct the run $\rho'$ on $\domof\tau$ in $\autof{\{S \disj S'\}}$ such that
$\forall w\in\domof\tau:\ \rho'(w)=S_1+S_2 $ where $\rho(w) = t_1+t_2, t_1\in
S_1 \wedge t_2\in S_2$. Since $\rho$ is accepting, $\rho'$ is accepting as well.
Therefore, $\tau\in\seman{\{S \disj S'\}}$.
\bigskip

$(b)$: Analogy to $(a)$.
\bigskip

$(c)$: From Lemma~\ref{lem:sem} we have that $\seman{\{ \pi_X(S) \}} =
\pi_X(\seman{S})$. We prove that $\pi_X(\seman{S}) = \seman{\{ \pi_X(t)~|~t\in S
\}}$.

$(\subseteq)$: Consider some $\tau\in\pi_X(\seman{S})$. Then, there is a tree
$\tau'\in\seman{S}$ such that $\tau\in\pi_X(\tau')$. Therefore, there is a
accepting run $\rho$ on $\tau'$ in $\autof{\expt{S}}$ and hence there is the
accepting run $\rho'$ on $\tau'$ in $\autof{\{ \pi_X(t)~|~t\in \expt{S} \}}$
defined as $\forall w\in\domof\tau:\ \rho'(w) = \pi_X(\rho(w))$ which implies
$\tau\in\seman{\{ \pi_X(t)~|~t\in S \}}$.

$(\supseteq)$: Consider $\tau\in\in\seman{\{ \pi_X(t)~|~t\in S \}}$. Therefore,
threre is an accepting run $\rho'$ on some $\tau'$ in $\autof{\expt{S}}$ defined
as $\forall w\in\domof\tau:\ \rho'(w) = t$ where $\rho(w) = \pi_X(t)$. Moreover,
we have $\tau\in\pi_X(\tau')$. Hence $\tau\in\pi_X(\seman{S})$.
\qed
\end{proof}

\medskip

\lemmaNondetUnion*

\begin{proof}
$(\subseteq)$: From Lemma~\ref{lem:sem} for a modified transition function, we
have $\langpof{\{ S \disj S' \}} = \langpof{S} \cup \langpof{S'}$. Now
assume $\tau\in \langpof{S} \cup \langpof{S'}$. Then, there is an accepting
run $\rho$ on $\tau$ either in $\autof{\expt{S}}$ or in $\autof{\expt{S'}}$.
Therefore, $\rho$ is an accepting run on $\tau$ also in $\autof{\expt{S}\cup
\expt{S'}}$.

$(\supseteq)$: We assume that $\tau\in\langpof{S \cup S'}$. Since for each
$t_1\in \expt{S}$ and $t_2\in \expt{S'}$ holds $t_1 \not\bowtie t_2$, we have
that $t\in\Delta'_a(t_1,t_2)$ is equal to $\emptyset$ (and vice versa).
Therefore, if $\rho$ is an accepting run on $\tau$ in $\autof{\expt{S}\cup
\expt{S'}}$, then $\rho$ is an accepting run in $\autof{\expt{S}}$ or in
$\autof{\expt{S'}}$. Hence, $\tau\in\langpof{\{ S \disj S' \}}$.
\qed
\end{proof}
\vfill
\eject

\vspace{-0.0mm}
\section{Basic WS$k$S Predicates}\label{app:wsks-pred}
\vspace{-0.0mm}
\begin{align}
	X = \emptyset &\Leftrightarrow \forall Y.\ X \subseteq Y \\
	\singof X &\Leftrightarrow \forall Y.\ Y \subseteq X \Rightarrow Y = X \wedge Y = \emptyset
\end{align}

\vspace{-0.0mm}
\section{Basic TAs}\label{sec:basicTA}
\vspace{-0.0mm}

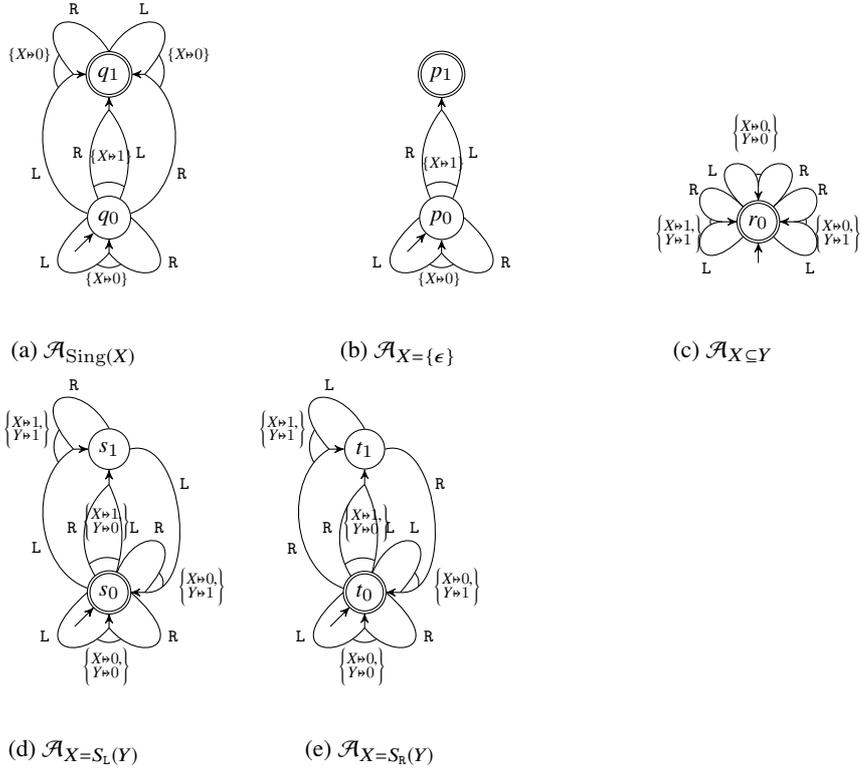
\begin{figure}[h]
  \begin{subfigure}[b]{0.3\linewidth}
  \begin{center}
    \begin{tikzpicture}[scale=0.95,transform shape]

\tikzstyle{state}=[draw,circle,inner sep=0.8mm]
\tikzstyle{hidnode}=[inner sep=0]
\tikzstyle{ref}=[inner sep=1mm]
\tikzstyle{lab}=[]

\tikzstyle{trans}=[->,>=stealth']
\tikzstyle{hidtrans}=[]
\tikzstyle{ark}=[]


\node[hidnode] (start2) at (2mm,3mm) {};
\node[state] (q0) [right of=start2,node distance=24mm] {$q_0$};
\node[state,accepting] (q1) [above of=q0,node distance=20mm] {$q_1$};

\node[hidnode,minimum size=0pt] (q1south) [below of=q1,node distance=5mm] {};
\node[hidnode,minimum size=0pt] (q1west) [left of=q1,node distance=5mm] {};
\node[hidnode,minimum size=0pt] (q1east) [right of=q1,node distance=5mm] {};
\node[hidnode,minimum size=0pt] (q2south) [below of=q0,node distance=5mm] {};

\node[hidnode] (r2loop_con1) [below of=q0,node distance=10mm,xshift=-4mm] {};
\node[hidnode] (r2loop_con2) [left of=q0,node distance=6mm,xshift=-6mm,yshift=-8mm] {};

\node[hidnode] (r2loop_con5) [below of=q0,node distance=10mm,xshift=4mm] {};
\node[hidnode] (r2loop_con6) [right of=q0,node distance=6mm,xshift=6mm,yshift=-8mm] {};

\node[hidnode] (r2loop_con3) [above of=q1,node distance=10mm,xshift=4mm] {};
\node[hidnode] (r2loop_con4) [right of=q1,node distance=6mm,xshift=6mm,yshift=8mm] {};

\node[hidnode] (r2loop_con7) [above of=q1,node distance=10mm,xshift=-4mm] {};
\node[hidnode] (r2loop_con8) [left of=q1,node distance=6mm,xshift=-6mm,yshift=8mm] {};


\draw[-] (q1.north)
	.. controls (r2loop_con3) and (r2loop_con4) ..
	coordinate[very near end] (qnew_t2_1)
        node[pos=0.35,above,yshift=-0.0mm,xshift=0mm] {$\scriptstyle\mathtt{L}$}
	(q1east.west);

\draw[-] (q0)
	edge[bend right=65]
	coordinate[very near end] (qnew_t2_2)
        node[pos=0.35,right,yshift=-0.0mm,xshift=0mm] {$\scriptstyle\mathtt{R}$}
	(q1east.west);

\draw[trans] (q1east) edge (q1.east);
\draw[ark] (qnew_t2_1) edge[bend left=30] (qnew_t2_2) node [right] {\scalebox{0.7}{$\unitrack{X}{0}$}};;


\draw[-] (q1.north)
	.. controls (r2loop_con7) and (r2loop_con8) ..
	coordinate[very near end] (qnew_t2_1)
        node[pos=0.35,above,yshift=-0.0mm,xshift=0mm] {$\scriptstyle\mathtt{R}$}
	(q1west.east);

\draw[-] (q0)
	edge[bend left=65]
	coordinate[very near end] (qnew_t2_2)
        node[pos=0.35,left,yshift=-0.0mm,xshift=0mm] {$\scriptstyle\mathtt{L}$}
	(q1west.east);

\draw[trans] (q1west) edge (q1.west);
\draw[ark] (qnew_t2_1) edge[bend right=30] (qnew_t2_2) node [left] {\scalebox{0.7}{$ \unitrack{X}{0}$}};;


\draw[-] (q0.east)
	.. controls (r2loop_con6) and (r2loop_con5) ..
	coordinate[very near end] (qnew_t2_1)
        node[pos=0.35,right,yshift=-0.0mm,xshift=0mm] {$\scriptstyle\mathtt{R}$}
	(q2south.north);

\draw[-] (q0.west)
	.. controls (r2loop_con2) and (r2loop_con1) ..
	coordinate[very near end] (qnew_t2_2)
        node[pos=0.35,left,yshift=-0.0mm,xshift=0mm] {$\scriptstyle\mathtt{L}$}
	(q2south.north);

\draw[trans] (q2south) edge (q0.south);
\draw[ark] (qnew_t2_1) edge[bend left=30] (qnew_t2_2) node [below,xshift=-2mm] {\scalebox{0.7}{$\unitrack{X}{0}$}};


\draw[-] (q0)
	edge[bend left=30]
	coordinate[very near start] (qnew_t2_1)
        node[pos=0.5,left] {$\scriptstyle\mathtt{R}$}
	(q1south.north);

\draw[-] (q0)
	edge[bend right=30]
	coordinate[very near start] (qnew_t2_2)
        node[pos=0.5,right] {$\scriptstyle\mathtt{L}$} (q1south.north);

\draw[trans] (q1south) edge (q1.south);
\draw[ark] (qnew_t2_1) edge[bend left=30] (qnew_t2_2) node [above,xshift=2.4mm,yshift=2mm] {\scalebox{0.7}{$\unitrack{X}{1}$}};;


\node[hidnode] (start3) [left of=q0, below of=q0,xshift=-0mm,node distance=5mm] {};
\draw[trans] (start3) edge (q0);

\end{tikzpicture}
  \end{center}
  \caption{$\A_{\singof X}$}
  \label{label}
  \end{subfigure}
  ~
  \begin{subfigure}[b]{0.3\linewidth}
  \begin{center}
    \begin{tikzpicture}[scale=0.95,transform shape]

\tikzstyle{state}=[draw,circle,inner sep=0.8mm]
\tikzstyle{hidnode}=[inner sep=0]
\tikzstyle{ref}=[inner sep=1mm]
\tikzstyle{lab}=[]

\tikzstyle{trans}=[->,>=stealth']
\tikzstyle{hidtrans}=[]
\tikzstyle{ark}=[]


\node[hidnode] (start2) at (2mm,3mm) {};
\node[state] (q0) [right of=start2,node distance=24mm] {$p_0$};
\node[state,accepting] (q1) [above of=q0,node distance=20mm] {$p_1$};

\node[hidnode,minimum size=0pt] (q1south) [below of=q1,node distance=5mm] {};
\node[hidnode,minimum size=0pt] (q2south) [below of=q0,node distance=5mm] {};

\node[hidnode] (r2loop_con1) [below of=q0,node distance=10mm,xshift=-4mm] {};
\node[hidnode] (r2loop_con2) [left of=q0,node distance=6mm,xshift=-6mm,yshift=-8mm] {};

\node[hidnode] (r2loop_con5) [below of=q0,node distance=10mm,xshift=4mm] {};
\node[hidnode] (r2loop_con6) [right of=q0,node distance=6mm,xshift=6mm,yshift=-8mm] {};


\draw[-] (q0.east)
	.. controls (r2loop_con6) and (r2loop_con5) ..
	coordinate[very near end] (qnew_t2_1)
        node[pos=0.35,right,yshift=-0.0mm,xshift=0mm] {$\scriptstyle\mathtt{R}$}
	(q2south.north);

\draw[-] (q0.west)
	.. controls (r2loop_con2) and (r2loop_con1) ..
	coordinate[very near end] (qnew_t2_2)
        node[pos=0.35,left,yshift=-0.0mm,xshift=0mm] {$\scriptstyle\mathtt{L}$}
	(q2south.north);

\draw[trans] (q2south) edge (q0.south);
\draw[ark] (qnew_t2_1) edge[bend left=30] (qnew_t2_2) node [below,xshift=-2mm] {\scalebox{0.7}{$\unitrack{X}{0}$}};


\draw[-] (q0)
	edge[bend left=30]
	coordinate[very near start] (qnew_t2_1)
        node[pos=0.5,left] {$\scriptstyle\mathtt{R}$}
	(q1south.north);

\draw[-] (q0)
	edge[bend right=30]
	coordinate[very near start] (qnew_t2_2)
        node[pos=0.5,right] {$\scriptstyle\mathtt{L}$} (q1south.north);

\draw[trans] (q1south) edge (q1.south);
\draw[ark] (qnew_t2_1) edge[bend left=30] (qnew_t2_2) node [above,xshift=2.4mm,yshift=1mm] {\scalebox{0.7}{$\unitrack{X}{1}$}};;


\node[hidnode] (start3) [left of=q0, below of=q0,xshift=-0mm,node distance=5mm] {};
\draw[trans] (start3) edge (q0);

\end{tikzpicture}
  \end{center}
  \caption{$\A_{X = \{ \epsilon \}}$}
  \label{label}
  \end{subfigure}
	~
	\begin{subfigure}[b]{0.3\linewidth}
  \begin{center}
    \begin{tikzpicture}[scale=0.95,transform shape]

\tikzstyle{state}=[draw,circle,inner sep=0.8mm]
\tikzstyle{hidnode}=[inner sep=0]
\tikzstyle{ref}=[inner sep=1mm]
\tikzstyle{lab}=[]

\tikzstyle{trans}=[->,>=stealth']
\tikzstyle{hidtrans}=[]
\tikzstyle{ark}=[]


\node[hidnode] (start2) at (2mm,3mm) {};
\node[state,accepting] (q0) [right of=start2,node distance=24mm] {$r_0$};

\node[hidnode,minimum size=0pt] (q1north) [above of=q0,node distance=5mm] {};
\node[hidnode,minimum size=0pt] (q1east) [right of=q0,node distance=5mm] {};
\node[hidnode,minimum size=0pt] (q1west) [left of=q0,node distance=5mm] {};

\node[hidnode] (r2loop_con1) [below of=q0,node distance=10mm,xshift=-4mm] {};
\node[hidnode] (r2loop_con2) [left of=q0,node distance=6mm,xshift=-0mm,yshift=-8mm] {};

\node[hidnode] (r2loop_con3) [left of=q0,node distance=8mm,xshift=-2mm] {};
\node[hidnode] (r2loop_con4) [left of=q0,above of=q0,node distance=6mm,xshift=-3mm,yshift=3mm] {};

\node[hidnode] (r2loop_con7) [left of=q0,node distance=8mm,xshift=-2mm] {};
\node[hidnode] (r2loop_con8) [left of=q0,below of=q0,node distance=6mm,xshift=-3mm,yshift=-3mm] {};

\node[hidnode] (r2loop_con5) [below of=q0,node distance=10mm,xshift=4mm] {};
\node[hidnode] (r2loop_con6) [right of=q0,node distance=6mm,xshift=0mm,yshift=-8mm] {};

\node[hidnode] (r2loop_con9) [right of=q0,node distance=8mm,xshift=2mm] {};
\node[hidnode] (r2loop_con10) [right of=q0,above of=q0,node distance=6mm,xshift=3mm,yshift=3mm] {};

\node[hidnode] (r2loop_con11) [right of=q0,node distance=8mm,xshift=2mm] {};
\node[hidnode] (r2loop_con12) [right of=q0,below of=q0,node distance=6mm,xshift=3mm,yshift=-3mm] {};

\node[hidnode] (r2loop_con13) [above of=q0,node distance=8mm,yshift=2mm] {};
\node[hidnode] (r2loop_con14) [right of=q0,above of=q0,node distance=6mm,xshift=3mm,yshift=3mm] {};

\node[hidnode] (r2loop_con15) [above of=q0,node distance=8mm,yshift=2mm] {};
\node[hidnode] (r2loop_con16) [left of=q0,above of=q0,node distance=6mm,xshift=-3mm,yshift=3mm] {};


\draw[-] (q0.north west)
	.. controls (r2loop_con4) and (r2loop_con3)  ..
	coordinate[very near end] (qnew_t2_1)
        node[pos=0.35,left,yshift=-0.0mm,xshift=0mm] {$\scriptstyle\mathtt{R}$}
	(q1west.east);

\draw[-] (q0.south west)
	.. controls (r2loop_con8) and (r2loop_con7) ..
	coordinate[very near end] (qnew_t2_2)
        node[pos=0.35,below,yshift=-0.0mm,xshift=0mm] {$\scriptstyle\mathtt{L}$}
	(q1west.east);

\draw[trans] (q1west) edge (q0.west);
\draw[ark] (qnew_t2_1) edge[bend right=30] (qnew_t2_2) node [left,yshift=-2mm] {\scalebox{0.7}{$\bintrack{X}{Y}{1}{1}$}};

\draw[-] (q0.north east)
	.. controls (r2loop_con10) and (r2loop_con9)  ..
	coordinate[very near end] (qnew_t2_1)
        node[pos=0.35,right,yshift=-0.0mm,xshift=0mm] {$\scriptstyle\mathtt{R}$}
	(q1east.west);

\draw[-] (q0.south east)
	.. controls (r2loop_con12) and (r2loop_con11) ..
	coordinate[very near end] (qnew_t2_2)
        node[pos=0.35,below,yshift=-0.0mm,xshift=0mm] {$\scriptstyle\mathtt{L}$}
	(q1east.west);

\draw[trans] (q1east) edge (q0.east);
\draw[ark] (qnew_t2_1) edge[bend left=30] (qnew_t2_2) node [right,yshift=-2mm] {\scalebox{0.7}{$\bintrack{X}{Y}{0}{1}$}};

\draw[-] (q0.north west)
	.. controls (r2loop_con16) and (r2loop_con15)  ..
	coordinate[very near end] (qnew_t2_1)
        node[pos=0.35,left,yshift=-0.0mm,xshift=0mm] {$\scriptstyle\mathtt{L}$}
	(q1north.south);

\draw[-] (q0.north east)
	.. controls (r2loop_con14) and (r2loop_con13)  ..
	coordinate[very near end] (qnew_t2_2)
        node[pos=0.35,right,yshift=-0.0mm,xshift=0mm] {$\scriptstyle\mathtt{R}$}
	(q1north.south);

\draw[trans] (q1north) edge (q0.north);
\draw[ark] (qnew_t2_1) edge[bend left=30] (qnew_t2_2) node [above,yshift=2mm] {\scalebox{0.7}{$\bintrack{X}{Y}{0}{0}$}};



\node[hidnode] (start3) [below of=q0,xshift=-0mm,node distance=6mm] {};
\draw[trans] (start3) edge (q0);

\end{tikzpicture}
  \end{center}
  \caption{$\A_{X \subseteq Y}$}
  \label{label}
  \end{subfigure}
	~
	\begin{subfigure}[b]{0.3\linewidth}
  \begin{center}
    \begin{tikzpicture}[scale=0.95,transform shape]

\tikzstyle{state}=[draw,circle,inner sep=0.8mm]
\tikzstyle{hidnode}=[inner sep=0]
\tikzstyle{ref}=[inner sep=1mm]
\tikzstyle{lab}=[]

\tikzstyle{trans}=[->,>=stealth']
\tikzstyle{hidtrans}=[]
\tikzstyle{ark}=[]


\node[hidnode] (start2) at (2mm,3mm) {};
\node[state,accepting] (q0) [right of=start2,node distance=24mm] {$s_0$};
\node[state] (q1) [above of=q0,node distance=20mm] {$s_1$};

\node[hidnode,minimum size=0pt] (q1south) [below of=q1,node distance=5mm] {};
\node[hidnode,minimum size=0pt] (q1west) [left of=q1,node distance=5mm] {};
\node[hidnode,minimum size=0pt] (q2east) [right of=q0,node distance=5mm,yshift=0mm] {};
\node[hidnode,minimum size=0pt] (q2south) [below of=q0,node distance=5mm] {};

\node[hidnode] (r2loop_con1) [below of=q0,node distance=10mm,xshift=-4mm] {};
\node[hidnode] (r2loop_con2) [left of=q0,node distance=6mm,xshift=-6mm,yshift=-8mm] {};

\node[hidnode] (r2loop_con5) [below of=q0,node distance=10mm,xshift=4mm] {};
\node[hidnode] (r2loop_con6) [right of=q0,node distance=6mm,xshift=6mm,yshift=-8mm] {};

\node[hidnode] (r2loop_con3) [above of=q0,node distance=10mm,xshift=4mm] {};
\node[hidnode] (r2loop_con4) [right of=q0,node distance=6mm,xshift=6mm,yshift=8mm] {};

\node[hidnode] (r2loop_con7) [above of=q1,node distance=10mm,xshift=-4mm] {};
\node[hidnode] (r2loop_con8) [left of=q1,node distance=6mm,xshift=-6mm,yshift=8mm] {};


\draw[-] (q1.east)
	edge[bend left=100]
	coordinate[very near end] (qnew_t2_1)
        node[pos=0.35,right,yshift=-0.0mm,xshift=0mm] {$\scriptstyle\mathtt{L}$}
	(q2east.west);

\draw[-] (q0)
	.. controls (r2loop_con3) and (r2loop_con4) ..
	coordinate[very near end] (qnew_t2_2)
        node[pos=0.35,above,right,yshift=2mm,xshift=0mm] {$\scriptstyle\mathtt{R}$}
	(q2east.west);

\draw[trans] (q2east) edge (q0.east);
\draw[ark] (qnew_t2_1) edge[bend right=30] (qnew_t2_2) node [right,xshift=1mm] {\scalebox{0.7}{$\bintrack{X}{Y}{0}{1}$}};;


\draw[-] (q1.north)
	.. controls (r2loop_con7) and (r2loop_con8) ..
	coordinate[very near end] (qnew_t2_1)
        node[pos=0.35,above,yshift=-0.0mm,xshift=0mm] {$\scriptstyle\mathtt{R}$}
	(q1west.east);

\draw[-] (q0)
	edge[bend left=65]
	coordinate[very near end] (qnew_t2_2)
        node[pos=0.35,left,yshift=-0.0mm,xshift=0mm] {$\scriptstyle\mathtt{L}$}
	(q1west.east);

\draw[trans] (q1west) edge (q1.west);
\draw[ark] (qnew_t2_1) edge[bend right=30] (qnew_t2_2) node [left] {\scalebox{0.7}{$ \bintrack{X}{Y}{1}{1}$}};;


\draw[-] (q0.east)
	.. controls (r2loop_con6) and (r2loop_con5) ..
	coordinate[very near end] (qnew_t2_1)
        node[pos=0.35,right,yshift=-0.0mm,xshift=0mm] {$\scriptstyle\mathtt{R}$}
	(q2south.north);

\draw[-] (q0.west)
	.. controls (r2loop_con2) and (r2loop_con1) ..
	coordinate[very near end] (qnew_t2_2)
        node[pos=0.35,left,yshift=-0.0mm,xshift=0mm] {$\scriptstyle\mathtt{L}$}
	(q2south.north);

\draw[trans] (q2south) edge (q0.south);
\draw[ark] (qnew_t2_1) edge[bend left=30] (qnew_t2_2) node [below,xshift=-2mm] {\scalebox{0.7}{$\bintrack{X}{Y}{0}{0}$}};


\draw[-] (q0)
	edge[bend left=40]
	coordinate[very near start] (qnew_t2_1)
        node[pos=0.5,left] {$\scriptstyle\mathtt{R}$}
	(q1south.north);

\draw[-] (q0)
	edge[bend right=20]
	coordinate[very near start] (qnew_t2_2)
        node[pos=0.5,right] {$\scriptstyle\mathtt{L}$} (q1south.north);

\draw[trans] (q1south) edge (q1.south);
\draw[ark] (qnew_t2_1) edge[bend left=30] (qnew_t2_2) node [above,xshift=2.4mm,yshift=2.1mm] {\scalebox{0.7}{$\bintrack{X}{Y}{1}{0}$}};


\node[hidnode] (start3) [left of=q0, below of=q0,xshift=-0mm,node distance=5mm] {};
\draw[trans] (start3) edge (q0);

\end{tikzpicture}
  \end{center}
  \caption{$\A_{X = \suckleftof Y}$}
  \label{label}
  \end{subfigure}
	~
	\begin{subfigure}[b]{0.3\linewidth}
  \begin{center}
    \begin{tikzpicture}[scale=0.95,transform shape]

\tikzstyle{state}=[draw,circle,inner sep=0.8mm]
\tikzstyle{hidnode}=[inner sep=0]
\tikzstyle{ref}=[inner sep=1mm]
\tikzstyle{lab}=[]

\tikzstyle{trans}=[->,>=stealth']
\tikzstyle{hidtrans}=[]
\tikzstyle{ark}=[]


\node[hidnode] (start2) at (2mm,3mm) {};
\node[state,accepting] (q0) [right of=start2,node distance=15mm] {$t_0$};
\node[state] (q1) [above of=q0,node distance=20mm] {$t_1$};

\node[hidnode,minimum size=0pt] (q1south) [below of=q1,node distance=5mm] {};
\node[hidnode,minimum size=0pt] (q1west) [left of=q1,node distance=5mm] {};
\node[hidnode,minimum size=0pt] (q2east) [right of=q0,node distance=5mm,yshift=0mm] {};
\node[hidnode,minimum size=0pt] (q2south) [below of=q0,node distance=5mm] {};

\node[hidnode] (r2loop_con1) [below of=q0,node distance=10mm,xshift=-4mm] {};
\node[hidnode] (r2loop_con2) [left of=q0,node distance=6mm,xshift=-6mm,yshift=-8mm] {};

\node[hidnode] (r2loop_con5) [below of=q0,node distance=10mm,xshift=4mm] {};
\node[hidnode] (r2loop_con6) [right of=q0,node distance=6mm,xshift=6mm,yshift=-8mm] {};

\node[hidnode] (r2loop_con3) [above of=q0,node distance=10mm,xshift=4mm] {};
\node[hidnode] (r2loop_con4) [right of=q0,node distance=6mm,xshift=6mm,yshift=8mm] {};

\node[hidnode] (r2loop_con7) [above of=q1,node distance=10mm,xshift=-4mm] {};
\node[hidnode] (r2loop_con8) [left of=q1,node distance=6mm,xshift=-6mm,yshift=8mm] {};


\draw[-] (q1.east)
	edge[bend left=100]
	coordinate[very near end] (qnew_t2_1)
        node[pos=0.35,right,yshift=-0.0mm,xshift=0mm] {$\scriptstyle\mathtt{R}$}
	(q2east.west);

\draw[-] (q0)
	.. controls (r2loop_con3) and (r2loop_con4) ..
	coordinate[very near end] (qnew_t2_2)
        node[pos=0.35,above,right,yshift=2mm,xshift=0mm] {$\scriptstyle\mathtt{L}$}
	(q2east.west);

\draw[trans] (q2east) edge (q0.east);
\draw[ark] (qnew_t2_1) edge[bend right=30] (qnew_t2_2) node [right,xshift=1mm] {\scalebox{0.7}{$\bintrack{X}{Y}{0}{1}$}};;


\draw[-] (q1.north)
	.. controls (r2loop_con7) and (r2loop_con8) ..
	coordinate[very near end] (qnew_t2_1)
        node[pos=0.35,above,yshift=-0.0mm,xshift=0mm] {$\scriptstyle\mathtt{L}$}
	(q1west.east);

\draw[-] (q0)
	edge[bend left=65]
	coordinate[very near end] (qnew_t2_2)
        node[pos=0.35,left,yshift=-0.0mm,xshift=0mm] {$\scriptstyle\mathtt{R}$}
	(q1west.east);

\draw[trans] (q1west) edge (q1.west);
\draw[ark] (qnew_t2_1) edge[bend right=30] (qnew_t2_2) node [left] {\scalebox{0.7}{$ \bintrack{X}{Y}{1}{1}$}};;


\draw[-] (q0.east)
	.. controls (r2loop_con6) and (r2loop_con5) ..
	coordinate[very near end] (qnew_t2_1)
        node[pos=0.35,right,yshift=-0.0mm,xshift=0mm] {$\scriptstyle\mathtt{R}$}
	(q2south.north);

\draw[-] (q0.west)
	.. controls (r2loop_con2) and (r2loop_con1) ..
	coordinate[very near end] (qnew_t2_2)
        node[pos=0.35,left,yshift=-0.0mm,xshift=0mm] {$\scriptstyle\mathtt{L}$}
	(q2south.north);

\draw[trans] (q2south) edge (q0.south);
\draw[ark] (qnew_t2_1) edge[bend left=30] (qnew_t2_2) node [below,xshift=-2mm] {\scalebox{0.7}{$\bintrack{X}{Y}{0}{0}$}};


\draw[-] (q0)
	edge[bend left=40]
	coordinate[very near start] (qnew_t2_1)
        node[pos=0.5,left] {$\scriptstyle\mathtt{R}$}
	(q1south.north);

\draw[-] (q0)
	edge[bend right=20]
	coordinate[very near start] (qnew_t2_2)
        node[pos=0.5,right] {$\scriptstyle\mathtt{L}$} (q1south.north);

\draw[trans] (q1south) edge (q1.south);
\draw[ark] (qnew_t2_1) edge[bend left=30] (qnew_t2_2) node [above,xshift=3mm,yshift=2mm] {\scalebox{0.7}{$\bintrack{X}{Y}{1}{0}$}};;


\node[hidnode] (start3) [left of=q0, below of=q0,xshift=-0mm,node distance=5mm] {};
\draw[trans] (start3) edge (q0);

\end{tikzpicture}
  \end{center}
  \caption{$\A_{X = \suckrightof Y}$}
  \label{label}
  \end{subfigure}
\caption{Tree automata of atomic WS2S formulae.}
\label{fig:basicTA}
\end{figure}

\end{document}